\documentclass[aps, pra, 10pt, superscriptaddress]{revtex4-1-npjQI}

\usepackage[utf8]{inputenc}
\usepackage[bookmarks]{hyperref}
\hypersetup{colorlinks=true,citecolor=blue,linkcolor=blue,filecolor=blue,urlcolor=black}
\usepackage{amsmath,amsfonts,amssymb,amsthm}
\usepackage{graphicx}
\usepackage[dvipsnames]{xcolor}
\usepackage{bbm} 
\usepackage{subfigure}
\usepackage{mathtools} 
\usepackage{enumerate}
\usepackage{lineno}
\usepackage{threeparttable}
\usepackage{multirow}   
\newcommand{\ket}[1]{\vert #1 \rangle}
\newcommand{\bra}[1]{\langle #1 \vert}
\newcommand\proj[1]{\vert #1 \rangle \langle #1 \vert}
\newcommand{\opn}[1]{\operatorname{#1}}
\DeclareMathOperator{\tr}{Tr}  

\newcommand*{\cD}{\mathcal{D}}
\newcommand*{\cS}{\mathcal{S}}
\newcommand{\1}{\mathbbm{1}}

\newcommand\summary[1]{\vspace{0.06in}\noindent\textbf{#1}}

\newtheorem{theorem}{Theorem}
\newtheorem{lemma}[theorem]{Lemma}

\usepackage{soul}
\setstcolor{green}
\usepackage[capitalize]{cleveref}

\newcommand{\sectionuser}[1]{{\bf #1 \\}}

\newcommand{\secfont}{\fontsize{11}{12}\usefont{OT1}{phv}{b}{n}}

\setcitestyle{super,open={},close={}}
\newcommand{\oncite}[1]{\scalebox{1.3}[1.3]{\raisebox{-0.80ex}{\cite{#1}}}}
\renewcommand{\figurename}{{\usefont{OT1}{phv}{b}{n}Fig.}\sffamily}
\renewcommand{\thefigure}{{\usefont{OT1}{phv}{b}{n}\arabic{figure}}}

\begin{document}

\title{Towards the standardization of quantum state verification using optimal strategies}

\newcommand{\physics}{National Laboratory of Solid-state Microstructures, School of Physics, \\
Collaborative Innovation Center of Advanced Microstructures, \\
State Key Laboratory for Novel Software Technology, \\
Department of Computer Science and Technology, \\
Nanjing University, Nanjing 210093, China}%

\newcommand{\sust}{Shenzhen Institute for Quantum Science and Engineering, Southern University of Science and Technology, Shenzhen 518055, China}%

\author{Xinhe Jiang}%
\thanks{These authors contributed equally to this work.}
\affiliation{\physics}%

\author{Kun Wang}%
\thanks{These authors contributed equally to this work.}
\affiliation{\sust}%

\author{Kaiyi Qian}%
\thanks{These authors contributed equally to this work.}
\affiliation{\physics}%

\author{Zhaozhong Chen}%
\thanks{These authors contributed equally to this work.}
\affiliation{\physics}%

\author{Zhiyu Chen}%
\affiliation{\physics}%

\author{Liangliang Lu}%
\affiliation{\physics}%

\author{Lijun Xia}%
\affiliation{\physics}%

\author{Fangmin Song}%
\affiliation{\physics}%

\author{Shining Zhu}%
\affiliation{\physics}%

\author{Xiaosong Ma}%
\email{xiaosong.ma@nju.edu.cn}%
\affiliation{\physics}%


\begin{abstract}
Quantum devices for generating entangled states have been extensively studied and widely used. As so, it becomes necessary to verify that these devices truly work reliably and efficiently as they are specified. Here, we experimentally realize the recently proposed two-qubit entangled state verification strategies using both local measurements (nonadaptive) and active feed-forward operations (adaptive) with a photonic platform. About 3283/536 number of copies ($N$) are required to achieve a 99\% confidence to verify the target quantum state for nonadaptive/adaptive strategies. These optimal strategies provide the Heisenberg scaling of the infidelity $\epsilon$ as a function of $N$ ($\epsilon${}$\sim${}$N^r$) with the parameter $r=-1$, exceeding the standard quantum limit with $r=-0.5$. We experimentally obtain the scaling parameter of $r=-0.88\pm$0.03 and $-0.78\pm$0.07 for nonadaptive and adaptive strategies, respectively. Our experimental work could serve as a standardized procedure for the verification of quantum states.
\end{abstract}
\maketitle%

\noindent{\secfont{INTRODUCTION}}
\vspace*{0.1cm}

\noindent Quantum state plays an important role in quantum information processing.~\cite{Nielsen2010} Quantum devices for creating quantum states are building blocks for quantum technology. Being able to verify these quantum states reliably and efficiently is an essential step towards practical applications of quantum devices.~\cite{paris2004quantum} Typically, a quantum device is designed to output some desired state $\rho$, but the imperfection in the device's construction and noise in the operations may result in the actual output state deviating from it to some random and unknown states $\sigma_i$. A standard way to distinguish these two cases is quantum state
tomography.~\cite{PhysRevLett.111.160406,PhysRevLett.105.150401,Haah:2016:STQ:2897518.2897585,Donnell:2016:EQT:2897518.2897544,Donnell:2017:EQT:3055399.3055454} However, this method is both time-consuming and computationally challenging.~\cite{2005Nature438.643H,2014NaPhoton8621C}
Non-tomographic approaches have also been proposed to accomplish the task,~\cite{PhysRevLett.94.060501,PhysRevLett.106.230501,PhysRevLett.107.210404,2015NatCommu6E8498A,PhysRevLett.115.220502,2016NatCo713251M,PhysRevX.8.021060,Badescu:2019:QSC:3313276.3316344}
yet these methods make some assumptions either on the quantum states or on the available operations. It is then natural to ask whether there exists an efficient non-tomographic approach to accomplish the task?

The answer is affirmative. Quantum state verification protocol checks the device's quality efficiently. Various studies have been explored using local measurements.~\cite{PhysRevLett.115.220502,morimae2017verification,PhysRevX.8.021060,takeuchi2019resource} Some earlier works considered the verification of maximally entangled states.~\cite{HayashiJPAMG2006,PhysRevA.74.062321.2006,hayashi2008statistical,HayashiNJP2009} In the context of hypothesis testing, optimal verification of maximally entangled state is proposed in ref.~\oncite{HayashiJPAMG2006}. Under the independent and identically distributed setting, Hayashi et al.~discussed the hypothesis testing of the entangled pure states.~\cite{HayashiNJP2009} In a recent work,~\cite{PhysRevLett.120.170502} Pallister et al.~proposed an optimal strategy to verify non-maximally entangled two-qubit pure states under locally projective and nonadaptive measurements. The locality constraint induces only a constant-factor penalty over the nonlocal strategies. Since then, numerous works have been done along this line of research,~\cite{zhu2019efficient1,zhu2019efficient,wang2019optimal,yu2019optimal,li2019efficient,Liu2019Efficient,
li2019optimal} targeting on different states and measurements. Especially, the optimal verification strategies under local operations and classical communication are proposed recently,~\cite{wang2019optimal,yu2019optimal,li2019efficient} which exhibit better efficency. We also remark related works by Dimi\'{c} et al.~\cite{Dimi2018} and Saggio et al.,~\cite{Saggio2019} in which they developed a generic protocol for efficient entanglement detection using local measurements and with an exponentially growing confidence versus the number of copies of the quantum state.

In this work, we report an experimental two-qubit state verification procedure using both optimal nonadaptive (local measurements) and adaptive (active feed-forward operations) strategies with an optical setup. Compared with previous works merely on minimizing the number of measurement settings,~\cite{PhysRevLett.117.210504,bavaresco2018measurements,friis2019entanglement} we also minimize the number of copies (i.e.,~coincidence counts in our experiment) required to verify the quantum state generated by the quantum device. We perform two tasks--Task A and Task B. With Task A, we obtain a fitting infidelity and the number of copies required to achieve a 99\% confidence to verify the quantum state. Task B is performed to estimate the confidence parameter $\delta$ and infidelity parameter $\epsilon$ versus the number of copies $N$. We experimentally compare the scaling of $\delta$-$N$ and $\epsilon$-$N$ by applying the nonadaptive strategy~\cite{PhysRevLett.120.170502} and adaptive strategy~\cite{wang2019optimal,yu2019optimal,li2019efficient} to the two-qubit states. With our methods, we obtain a comprehensive judgement about the quantum state generated by a quantum device. Present experimental and data analysis workflow may be regarded as a standard procedure for quantum state verification.

\vspace*{0.5cm}
\noindent{\secfont{RESULTS}}
\vspace*{0.1cm}

\noindent\sectionuser{Quantum state verification}\quad
Consider a quantum device $\cD$ designed to produce the two-qubit pure state
\begin{equation}\label{eq:verify-state}
    \ket{\Psi} = \sin\theta\ket{HH} + \cos\theta\ket{VV},
\end{equation}
where $\theta\in[0,\pi/4]$. However, it might work incorrectly and actually outputs independent two-qubit fake states $\sigma_1,\sigma_2,\cdots,\sigma_N$ in $N$ runs. The goal of the verifier is to determine the fidelity threshold of these fake states to the target state with a certain confidence. We remark that the state for $\theta=\pi/4$ is the maximally entangled state and $\theta=0$ is the product state. As special cases of the general state in Eq.~\eqref{eq:verify-state}, all the analysis methods presented in the following can be applied to the verification of maximally entangled state and product state. The details of the verification strategies for maximally entangled state and product state are given in Supplementary Notes 1.C and 1.D. Previously, theoretical~\cite{HayashiJPAMG2006,HayashiNJP2009,zhu2019optimal} and experimental~\cite{PhysRevA.74.062321.2006} works have studied the verification of maximally entangled state. Here, we focus mainly on the verification of non-maximally entangled state in the main text, which is more advantageous in certain experiments comparing to maximally entangled state. For instance, in the context of loophole-free Bell test, non-maximally entangled states require lower detection efficiency than maximally entangled states~\cite{Eberhard1993,GiustinaNature2013,GiustinaPRL2015,ShalmPRL2015}. The details and experimental results for the verification of maximally entangled state and product state are shown in the Supplementary Notes 2 and 4. To realize the verification of our quantum device, we perform the following two tasks in our experiment (see Fig.~1):
\begin{description}
	\item[\textbf{Task A}] Performing measurements on the fake states copy by copy according to verification strategy, and make statistics on the number of copies required before we find the first fail event. The concept of Task A is shown in Fig.~1b.
	\item[\textbf{Task B}] Performing a fixed number ($N$) of measurements according to verification strategy, and make statistics on the number of copies that pass the verification tests. The concept of Task B is shown in Fig.~1c.
\end{description}

\begin{figure}[!htbp]
  \centering
  	\includegraphics[width=0.5\textwidth]{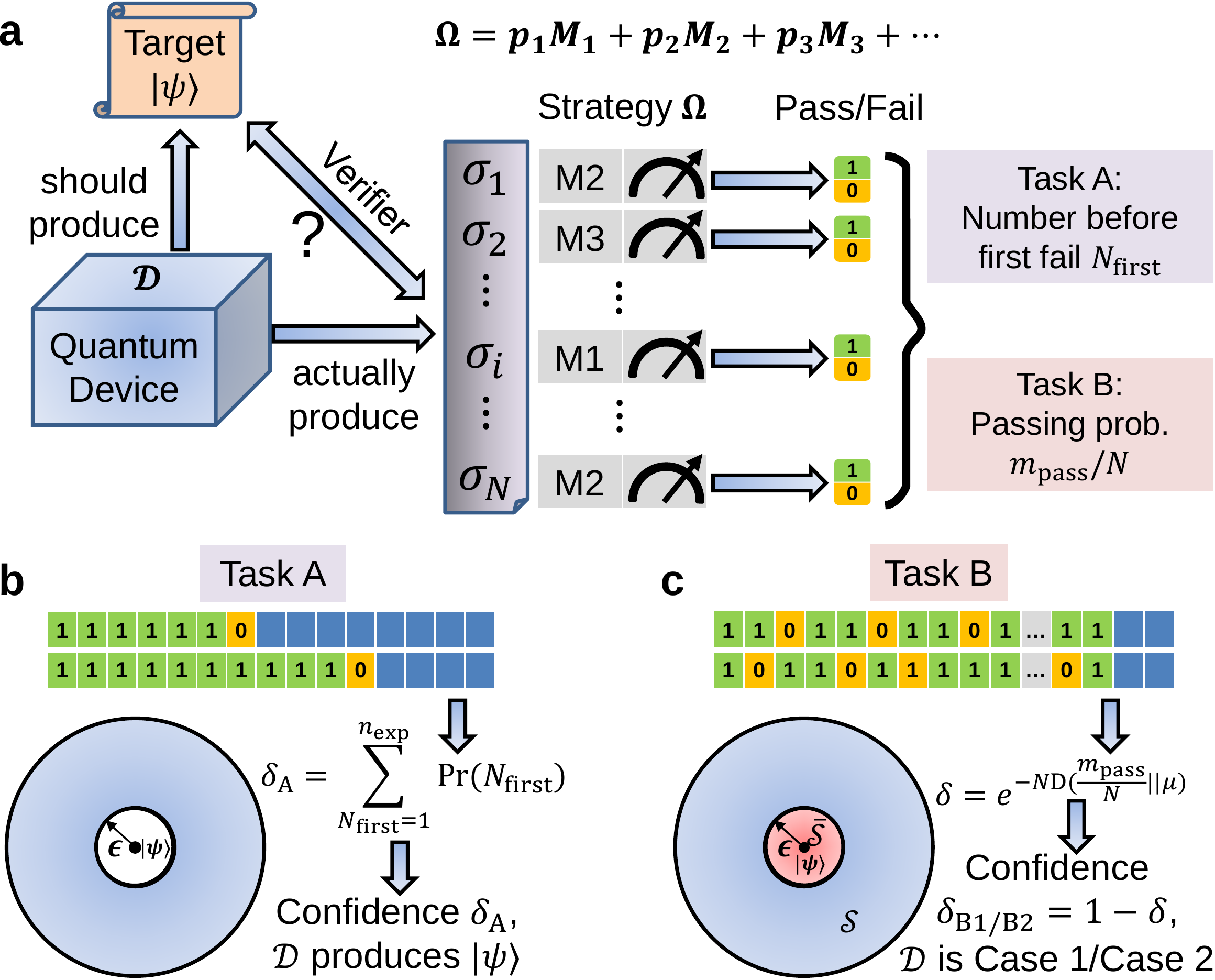}
  \caption{Illustration of quantum state verification strategy. \textbf{a}~Consider a quantum device $\cD$ designed to produce the two qubit pure state $\ket{\psi}$. However, it might work incorrectly and actually outputs two-qubit fake states $\sigma_1,\sigma_2,\cdots,\sigma_N$ in $N$ runs. For each copy $\sigma_i$, randomly projective measurements $\{M_1, M_2, M_3, \cdots\}$ are performed by the verifier based on their corresponding probabilities $\{p_1,p_2,p_3,\cdots\}$. Each measurement outputs a binary outcome 1 for pass and 0 for fail. The verifier takes two tasks based on these measurement outcomes. \textbf{b}~Task A gives the statistics on the number of copies required before finding the first fail event. From these statistics, the verifier obtains the confidence $\delta_\text{A}$ that the device outputs state $\ket{\psi}$. \textbf{c}~Task B performs a fixed number ($N$) of measurements and makes a statistic on the number of copies ($m_\text{pass}$) passing the test. From these statistics, the verifier can judge with a certain confidence $\delta_\text{B1}$/$\delta_\text{B2}$ that the device belongs to Case 1 or Case 2.}
  \label{fig1:Principle}
\end{figure}
\textbf{Task A} is based on the assumption that there exists some $\epsilon>0$ for which the fidelity $\langle\Psi\vert\sigma_i\vert\Psi\rangle$ is either $1$ or satisfies $\langle\Psi\vert\sigma_i\vert\Psi\rangle\leq1-\epsilon$ for all $i\in\{1,\cdots,N\}$ (see Fig.~1b). Our task is to determine which is the case for the quantum device. To achieve \textbf{Task A}, we perform binary-outcome measurements from a set of available projectors to test the state. Each binary-outcome measurement $\{M_l, \1-M_l\}$ ($l=1,2,3,\cdots$) is specified by an operator $M_l$, corresponding to passing the test. For simplicity, we use $M_l$ to denote the corresponding binary measurement. This measurement is performed with probability $p_l$. We require the target state $\ket{\Psi}$ always passes the test, i.e.,~$M_l\ket{\Psi}=\ket{\Psi}$. In the bad case ($\langle\Psi\vert\sigma_i\vert\Psi\rangle\leq1-\epsilon$), the maximal probability that $\sigma_i$ can pass the test is given by~\cite{PhysRevLett.120.170502,zhu2019efficient1}
\begin{equation}\label{eq:fail-probability}
  \max_{\bra{\Psi}\sigma_i\ket{\Psi}\leq 1-\epsilon }\tr(\Omega \sigma_i)
= 1 - [1-\lambda_2(\Omega)]\epsilon := 1-\Delta_\epsilon,
\end{equation}
where $\Omega=\sum_l p_l M_l$ is called an strategy, $\Delta_\epsilon$ is the probability $\sigma_i$ fails a test and $\lambda_2(\Omega)$ is the second largest eigenvalue of
$\Omega$. Whenever $\sigma_i$ fails the test, we know immediately that the device works incorrectly. After $N$ runs,
$\sigma_i$ in the incorrect case can pass all these tests with probability being at most
$[1-[1-\lambda_2(\Omega)]\epsilon]^N$. Hence to achieve confidence $1 - \delta$, it suffices to conduct $N$ number of
measurements satisfying~\cite{PhysRevLett.120.170502}
\begin{equation}\label{eq:num-measurements}
N \geq \frac{\ln\delta}{\ln[1-[1-\lambda_2(\Omega)]\epsilon]}
\approx \frac{1}{[1-\lambda_2(\Omega)]\epsilon}\ln\frac{1}{\delta}.
\end{equation}

From Eq.~\eqref{eq:num-measurements} we can see that an optimal strategy is obtained by minimizing the second largest
eigenvalue $\lambda_2(\Omega)$, with respect to the set of available measurements. Pallister et al.~\cite{PhysRevLett.120.170502} proposed an optimal strategy for \textbf{Task A}, using only locally projective
measurements. Since no classical communication is involved, this strategy (hereafter labelled as $\Omega_{\opn{opt}}$) is nonadaptive. Later, Wang et al.~\cite{wang2019optimal}, Yu et al.~\cite{yu2019optimal} and Li et al.~\cite{li2019efficient} independently propose the optimal strategy using one-way local operations and classical communication (hereafter labelled as $\Omega^\rightarrow_{\opn{opt}}$) for two-qubit pure states. Furthermore, Wang et al.~\cite{wang2019optimal} also gives the optimal strategy for two-way classical communication. The adaptive strategy allows general local operations and classical communication measurements and is shown to be more efficient than the strategies based on local measurements. Thus it is important to realize the adaptive strategy in the experiment. We refer to the Supplementary Notes 1 and 2 for more details on these strategies.


In reality, quantum devices are never perfect. Another practical scenario is to conclude with high confidence that the fidelity of the output states are above or below a certain threshold. To be specific, we want to
distinguish the following two cases:
\begin{description}
	\item[\textbf{Case 1}] $\cD$ works correctly -- $\forall i,\bra{\psi}\sigma_i\ket{\psi} > 1-\epsilon$. In this case, we regard the device as ``good''.
	\item[\textbf{Case 2}] $\cD$ works incorrectly -- $\forall i,\bra{\psi}\sigma_i\ket{\psi} \leq 1-\epsilon$. In this case, we regard the device as ``bad''.
\end{description}
We call this \textbf{Task B} (see Fig.~1c), which is different from \textbf{Task A}, since the condition for
`$\cD$ works correctly' is less restrictive compared with that of \textbf{Task A}. It turns out that the verification strategies proposed for \textbf{Task A} are readily applicable to \textbf{Task B}. Concretely, we perform the
nonadaptive verification strategy $\Omega_{\opn{opt}}$ sequentially in $N$ runs and count the number of passing
events $m_{\text{pass}}$. Let $X_i$ be a binary variable corresponding to the event that $\sigma_i$ passes the test ($X_i=1$)
or not ($X_i=0$). Thus we have $m_{\text{pass}} = \sum_{i=1}^NX_i$.
Assume that the device is ``good'', then from Eq.~\eqref{eq:fail-probability} we can derive that the passing
probability of the generated states is no smaller than $1 - [1-\lambda_2(\Omega_{\opn{opt}})]\epsilon$. We refer to Lemma 3 in the
Supplementary Note 3.A for proof. Thus the expectation of $X_i$ satisfies $\mathbb{E}[X_i]\geq1
- (1-\lambda_2(\Omega_{\opn{opt}}))\epsilon \equiv
\mu$. The independence assumption together with the law of large numbers then
guarantee $m_{\text{pass}}\geq N\mu$, when $N$ is sufficiently large.
We follow the statistical analysis methods using the Chernoff bound in the context of state verification~\cite{Dimi2018,Saggio2019,yu2019optimal,zhang2019experimental}, which is related to the security analysis of quantum key distributions~\cite{ScaraniRMP2009,Hayashi2014}. We then upper bound the probability that the device works incorrectly as
\begin{equation}\label{Large}
\delta\equiv e^{-N\operatorname{D}\left(\frac{m_{\text{pass}}}{N}\middle\lVert\mu\right)},
\end{equation}
where $\operatorname{D}\left(x\lVert y\right):=x\log_2\frac{x}{y}+(1-x)\log_2\frac{1-x}{1-y}$ is the
Kullback-Leibler divergence. That is to say, we can conclude with confidence $\delta_\text{B1}=1-\delta$ that $\cD$
belongs to \textbf{Case 1}. Conversely, if the device is ``bad'', then using the same argument we can
conclude with confidence $\delta_\text{B2}=1-\delta$ that $\cD$ belongs to \textbf{Case 2}. Please refer to the Supplementary Note 3 for rigorous proofs and arguments on how to evaluate the performance of the quantum device for these two cases.

To perform \textbf{Task B} with the adaptive strategy $\Omega^\rightarrow_{\opn{opt}}$, we record the number of passing events $m_\text{pass} = \sum_{i=1}^NX_i$. If the device is ``good'', the passing probability of the generated states is no smaller than $\mu_s \equiv 1 - [1-\lambda_4(\Omega^\rightarrow_{\opn{opt}})]\epsilon$ where $\lambda_4(\Omega^\rightarrow_{\opn{opt}})=\sin^2\theta/(1+\cos^2\theta)$ is the smallest eigenvalue of $\Omega^\rightarrow_{\opn{opt}}$, as proved by Lemma 5 in Supplementary Note 3.B. The independence assumption along with the law of large numbers guarantee that $m_{\text{pass}}\geq N\mu_s$, when $N$ is sufficiently large. On the other hand, if the device is ``bad'', we can prove that the passing probability of the generated states is no larger than $\mu_l \equiv 1 - [1-\lambda_2(\Omega^\rightarrow_{\opn{opt}})]\epsilon$, where $\lambda_2(\Omega^\rightarrow_{\opn{opt}})=\cos^2\theta/(1+\cos^2\theta)$, by Lemma 4 in Supplementary Note 3.B. Again, the independence assumption and the law of large numbers guarantee that $m_{\text{pass}}\leq N\mu_l$, when $N$ is large enough. Therefore, we consider two regions regarding the value of $m_{\text{pass}}$ in the adaptive strategy, i.e.,~the region $m_{\text{pass}}\leq N\mu_s$ and the region $m_{\text{pass}}\geq N\mu_l$. In these regions, we can conclude with $\delta_\text{B1}=1-\delta_l$/$\delta_\text{B2}=1-\delta_s$ that the device belongs to \textbf{Case 1}/\textbf{Case 2}. The expressions for $\delta_l$ and $\delta_s$ and all the details for applying adaptive strategy to \textbf{Task B} can be found in Supplementary Note 3.B.

\vspace*{0.5cm}
\noindent\sectionuser{Experimental setup and verification procedure}\quad
Our two-qubit entangled state is generated based on a type-II spontaneous parametric down-conversion in a 20 mm-long periodically-poled potassium titanyl phosphate (PPKTP) crystal, embedded in a Sagnac interferometer~\cite{Kim2006PRA.73.012316,Fedrizzi2007OE.15.15377} (see Fig.~2). A continuous-wave external-cavity ultraviolet (UV) diode laser at 405 nm is used as the pump light. A half-wave plate (HWP1) and quarter-wave plate (QWP1) transform the linear polarized light into the appropriate elliptically polarized light to provide the power balance and phase control of the pump field. With an input pump power of $\sim$30 mW, we typically obtain 120 kHz coincidence counts.

\begin{figure}[!htbp]
  \centering
  	\includegraphics[width=0.5\textwidth]{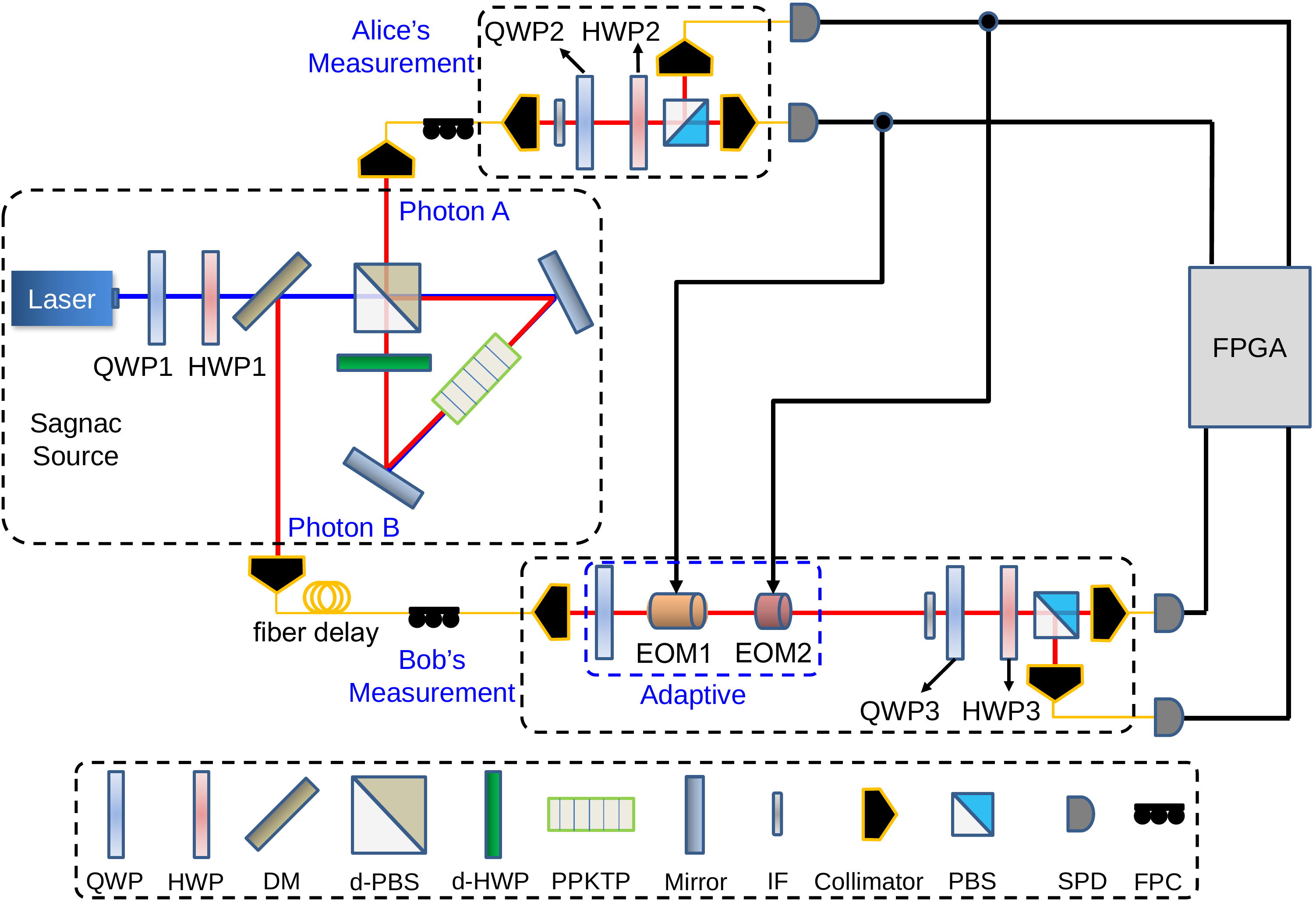}
  \caption{Experimental setup for optimal verification of two-qubit quantum state. We use a photon pair source based on a Sagnac interferometer to generate various two-qubit quantum state. QWP1 and HWP1 are used for adjusting the relative amplitude of the two counter-propagating pump light. For nonadaptive strategy, the measurement is realized with QWP, HWP and polarizing beam splitter (PBS) at both Alice's and Bob's site. The adaptive measurement is implemented by real-time feed-forward operation of electro-optic modulators (EOMs), which are triggered by the detection signals recorded with a field-programmable gate array (FPGA). The optical fiber delay is used to compensate the electronic delay from Alice's single photon detector (SPD) to the two EOMs. QWP: quarter-wave plate; HWP: half-wave plate; DM: Dichroic mirror; PBS: polarizing beam splitter; IF: 3-nm interference filter centered at 810 nm; dPBS: dual-wavelength polarizing beam splitter; dHWP: dual-wavelength half-wave plate; PPKTP: periodically poled KTiOPO$_4$; FPC: Fiber polarization controller.}
  \label{fig2:Sagnac}
\end{figure}

The target state has the following form
\begin{equation}\label{eq:target}
	\ket{\psi} = \sin\theta\ket{HV} + e^{i\phi}\cos\theta\ket{VH},
\end{equation}
where $\theta$ and $\phi$ represent amplitude and phase, respectively. This state is locally equivalent to $\ket{\Psi}$ in Eq.~\eqref{eq:verify-state} by $ \mathbb{U} = \begin{pmatrix} 1 & 0 \\ 0 & 1 \end{pmatrix} \otimes \begin{pmatrix} 0 &  e^{i\phi} \\ 1 & 0 \end{pmatrix}$. By using Lemma 1 in Supplementary Note 1, the optimal strategy for verifying $\ket{\psi}$ is $\Omega_{\opn{opt}}'$=$\mathbb{U}\Omega_{\opn{opt}}\mathbb{U}^\dagger$, where $\Omega_{\opn{opt}}$ is the optimal strategy verifying $\ket{\Psi}$ in Eq.~\eqref{eq:verify-state}. In the Supplementary Note 2, we write down explicitly the optimal nonadaptive strategy~\cite{PhysRevLett.120.170502} and adaptive strategy~\cite{wang2019optimal,yu2019optimal,li2019efficient} for verifying $\ket{\psi}$.

In our experiment, we implement both the nonadaptive and adaptive measurements to realize the verification strategies. There are four settings $\{P_0, P_1, P_2, P_3\}$ for nonadaptive measurements~\cite{PhysRevLett.120.170502} while only three settings $\{\widetilde{T}_0, \widetilde{T}_1,
\widetilde{T}_2\}$ are required for the adaptive
measurements.~\cite{wang2019optimal,yu2019optimal,li2019efficient} The exact form of these projectors is given in the
Supplementary Note 2. Note that the measurements $P_0=\widetilde{T}_0=\proj{H}\otimes\proj{V}+\proj{V}\otimes\proj{H}$
are determined by the standard $\sigma_z$ basis for both the nonadaptive and adaptive strategies, which are orthogonal and can be
realized with a combination of QWP, HWP and polarization beam splitter (PBS). For adaptive measurements, the measurement
bases $\widetilde{v}_{+} = e^{i\phi}\cos\theta\ket{H}+\sin\theta\ket{V}$ / $\widetilde{w}_{+} = e^{i\phi}\cos\theta\ket{H}-i\sin\theta\ket{V}$ and
$\widetilde{v}_{-} = e^{i\phi}\cos\theta\ket{H}-\sin\theta\ket{V}$ / $\widetilde{w}_{-} = e^{i\phi}\cos\theta\ket{H}+i\sin\theta\ket{V}$ at Bob's site are not orthogonal. Note that we only implement the one-way adaptive strategy in our experiment. The two-way adaptive strategy is also derived in ref.~\oncite{wang2019optimal}. Compared to nonadaptive and one-way adaptive strategy, the two-way adaptive strategy gives improvements on the verification efficiency due to the utilization of more classical communication resources. The implementation of two-way adaptive strategy requires: First, Alice performs her measurement and sends her results to Bob; Then, Bob performs his measurement according to Alice's outcomes; Finally, Alice performs another measurement conditioning on Bob's measurement outcomes. This procedure requires the real-time communications both from Alice to Bob and from Bob to Alice. Besides, the two-way adaptive strategy requires the quantum nondemolition measurement at Alice's site, which is difficult to implement in the current setup. To realize the one-way adaptive strategy, we transmit the results of Alice's measurements to Bob through classical communication channel, which is implemented by real-time feed-forward operations of the electro-optic modulators (EOMs). As shown in
Fig.~2, we trigger two EOMs at Bob's site to realize the adaptive measurements based on the results
of Alice's measurement. If Alice's outcome is $\ket{+}=(\ket{V}+\ket{H})/\sqrt{2}$ or
$\ket{R}=(\ket{V}+i\ket{H})/\sqrt{2}$, EOM1 implement the required rotation and EOM2 is identity operation. Conversely, if Alice's
outcome is $\ket{-}=(\ket{V}-\ket{H})/\sqrt{2}$ or $\ket{L}=(\ket{V}-i\ket{H})/\sqrt{2}$, EOM2 will implement the required rotation and
EOM1 is identity operation. Our verification procedure is the following.

(1) Specifications of quantum device. We adjust the HWP1 and QWP1 of our Sagnac source to generate the desired quantum state.

(2) Verification using the optimal strategy. In this stage, we generate many copies of the quantum state sequentially with our Sagnac source. These copies are termed as fake states $\{\sigma_i, i=1,2,\cdots,N\}$. Then, we perform the optimal nonadaptive verification strategy to $\sigma_i$. From the parameters $\theta$ and $\phi$ of target state, we can compute the angles of wave plates QWP2 and HWP2, QWP3 and HWP3 for realizing the projectors $\{P_0, P_1, P_2, P_3\}$ required in the nonadaptive strategy. To implement the adaptive strategy, we employ two EOMs to realize the $\widetilde{v}_{+}$/$\widetilde{v}_{-}$ and $\widetilde{w}_{+}$/$\widetilde{w}_{-}$ measurements once receiving Alice's results (refer to Supplementary Note 2.B for the details). Finally, we obtain the timetag data of the photon detection from the field programmable gate array (FPGA) and extract individual coincidence count (CC) which is regarded as one copy of our target state. We use the timetag experimental technique to record the channel and arrival time of each detected photon for data processing.~\cite{timetagmanual} The time is stored as multiples of the internal time resolution ($\sim$156 ps). The first data in the timetag is recorded as the starting time $t_{i0}$. With the increasing of time, we search the required CC between different channels within a fixed coincidence window (0.4 ns). If a single CC is obtained, we record the time of the ended timetag data as $t_{f0}$. Then we move to the next time slice $t_{i1}$--$t_{f1}$ to search for the next CC. This process can be cycled until we find the $N$-th CC in time slice $t_{iN-1}$--$t_{fN-1}$. This measurement can be viewed as single-shot measurement of the bipartite state with post-selection. The time interval in each slice is about 100 $\mu$s in our experiment, consistent with the 1/CR, CR-coincidence rate. By doing so, we can precisely obtain the number of copies $N$ satisfying the verification requirements. We believe this procedure is suitable in the context of verification protocol, because one wants to verify the quantum state with the minimum amount of copies.

(3) Data processing. From the measured timetag data, the results for different measurement settings can be obtained.
For the nonadaptive strategy, $\{P_0, P_1, P_2, P_3\}$ are chosen randomly with the probabilities \{$\mu_0$, $\mu_1$, $\mu_2$, $\mu_3$\}
($\mu_0$=$\alpha(\theta)$, $\mu_i$=(1-$\alpha(\theta)$)/3)) with $\alpha(\theta)=(2-\sin(2\theta))/(4+\sin(2\theta))$.
For the adaptive strategy, \{$\widetilde{T}_0$, $\widetilde{T}_1$, $\widetilde{T}_2$\} projectors are randomly chosen according to the
probabilities \{${\beta(\theta), (1-\beta(\theta))/2, (1-\beta(\theta))/2}$\}, where $\beta(\theta)=\cos^2\theta/(1+\cos^2\theta)$. For \textbf{Task A}, we use CC to decide whether the outcome of each measurement is pass or fail for each $\sigma_i$. The passing
probabilities for the nonadaptive strategy can be, respectively, expressed as,
\begin{eqnarray}
	P_0: \frac{CC_{HV}+CC_{VH}}{CC_{HH}+CC_{HV}+CC_{VH}+CC_{VV}}, \\
	P_i: \frac{CC_{\widetilde{u}_i \widetilde{v}_i^\perp}+CC_{\widetilde{u}_i^\perp \widetilde{v}_i}+CC_{\widetilde{u}_i^\perp \widetilde{v}_i^\perp}}{CC_{\widetilde{u}_i \widetilde{v}_i}+CC_{\widetilde{u}_i \widetilde{v}_i^\perp}+CC_{\widetilde{u}_i^\perp \widetilde{v}_i}+CC_{\widetilde{u}_i^\perp \widetilde{v}_i^\perp}}.
\end{eqnarray}
where $i=1,2,3$, and $\widetilde{u}_i/\widetilde{u}_i^\perp$ and $\widetilde{v}_i/\widetilde{v}_i^\perp$ are the
orthogonal bases for each photon and their expressions are given in the Supplementary Note 2.A.
For $P_0$, if the individual CC is in $CC_{HV}$ or $CC_{VH}$, it indicates that $\sigma_i$ passes the test and we
set $X_i=1$; otherwise, it fails to pass the test and we set $X_i=0$. For $P_i$, $i=1,2,3$, if the individual CC is in
$CC_{\widetilde{u}_i\widetilde{v}_i^\perp}$, $CC_{\widetilde{u}_i^\perp \widetilde{v}_i}$ or $CC_{\widetilde{u}_i^\perp
\widetilde{v}_i^\perp}$, it indicates that $\sigma_i$ passes the test and we set $X_i=1$;
otherwise, it fails to pass the test and we set $X_i=0$. For the adaptive strategy, we set the value of the random variables $X_i$ in a similar way.

We increase the number of copies ($N$) to decide the occurrence of the first failure for \textbf{Task
A} and the frequency of passing events for \textbf{Task B}. From these data, we obtain the relationship of the confidence parameter $\delta$, the infidelity parameter $\epsilon$, and the number of copies $N$. There are certain probabilities that the verifier fail for each
measurement. In the worst case, the probability that the verifier fails to assert $\sigma_i$ is given by
$1-\Delta_\epsilon$, where $\Delta_\epsilon=1-\epsilon/(2+\sin\theta\cos\theta)$ for nonadaptive
strategy~\cite{PhysRevLett.120.170502} and $\Delta_\epsilon=1-\epsilon/(2-\sin^2\theta)$ for adaptive strategy.~\cite{wang2019optimal,yu2019optimal,li2019efficient}

\vspace*{0.5cm}
\noindent\sectionuser{Results and analysis of two-qubit optimal verification}\quad
The target state to be verified is the general two-qubit state in Eq.~\eqref{eq:target}, where the parameter $\theta=k*\pi/10$ and $\phi$ is optimized with maximum likelihood estimation method. In this section, we present the results of $k=2$ state (termed as k2, see Supplementary Note 2) as an example. The verification results of other states, such as the maximally entangled state and the product state, are presented in Supplementary Note 4. Our theoretical non-maximally target state is specified by $\theta=0.6283$ ($k=2$). In experiment, we obtain $\ket{\psi}=0.5987\ket{HV}+0.8010e^{3.2034i}\ket{VH}$ ($\theta=0.6419$, $\phi=3.2034$) as our target state to be verified. In order to realize the verification strategy, the projective measurement is performed sequentially by randomly choosing the projectors. We take 10000 rounds for a fixed 6000 number of copies.

\textbf{Task A}. According to this verification task, we make a statistical analysis on the number of measurements required for the first occurrence of failure. According to the geometric distribution, the probability that the $n$-th measurement (out of $n$ measurements) is the first failure is
\begin{equation}\label{eq:GeoDis}
	\text{Pr}(N_\text{first}=n)=(1-\Delta_\epsilon)^{n-1}{\cdot}\Delta_\epsilon
\end{equation}
where $n=1,2,3,\cdots$. We then obtain the cumulative probability
\begin{equation}\label{eq:deltaA}
	\delta_\text{A}=\sum_{N_\text{first}=1}^{n_\text{exp}}\text{Pr}(N_\text{first})
\end{equation}
which is the confidence of the device generating the target state $\ket{\psi}$. In Fig.~3a, we show the distribution of the number $N_\text{first}$ required before the first failure for the nonadaptive (Non) strategy. From the figure we can see that $N_\text{first}$ obeys the geometric distribution. We fit the distribution with the function in Eq.~\eqref{eq:GeoDis} and obtain an experimental infidelity $\epsilon^\text{Non}_\text{exp}=0.0034(15)$, which is a quantitative estimation of the infidelity for the generated state. From the experimental statistics, we obtain the number $n_\text{exp}^\text{Non}$=3283 required to achieve the 99\% confidence (i.e.,~99\% cumulative probability for $N_\text{first}\leq n_\text{exp}^\text{Non}$) of judging the generated states to be the target state in the nonadaptive strategy.

\begin{figure}[!htbp]
  \centering
  	\includegraphics[width=0.6\textwidth]{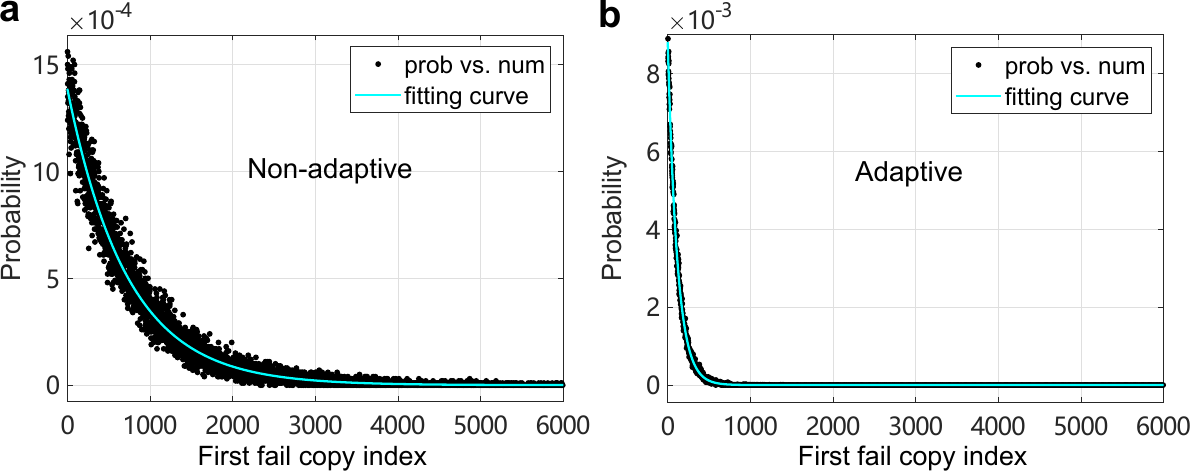}
  \caption{The distribution of the number required before the first failure. \textbf{a} for the nonadaptive strategy. \textbf{b} for the adaptive strategy. From the statistics, we obtain the fitting infidelity of $\epsilon^\text{Non}_\text{exp}=0.0034(15)$ and $\epsilon^\text{Adp}_\text{exp}=0.0121(6)$. The numbers required to achieve a 99\% confidence are $n_\text{exp}^\text{Non}$=3283 and $n_\text{exp}^\text{Adp}$=536, respectively.}
  \label{fig3:FirstFail}
\end{figure}

The results for the adaptive (Adp) verification of \textbf{Task A} are shown in Fig.~3b. The experimental fitting infidelity for this distribution is $\epsilon^\text{Adp}_\text{exp}=0.0121(6)$. The number required to achieve the same 99\% confidence as the nonadaptive strategy is $n_\text{exp}^\text{Adp}$=536. Note this nearly six times (i.e.,~$n_\text{exp}^\text{Non}/n_\text{exp}^\text{Adp}\sim6$) difference of the experimental number required to obtain the 99\% confidence is partially because the infidelity with adaptive strategy is approximately four times larger than the nonadaptive strategy. However, the number of copies required to achieve the same confidence by using the adaptive strategy is still about two times fewer than the nonadaptive strategy even if the infidelity of the generated states is the same, see the analysis presented in Supplementary Note 5. This indicates that the adaptive strategy requires a significant lower number of copies to conclude the device output state $\ket{\psi}$ with 99\% confidence compared with the nonadaptive one.

\begin{figure}[!htbp]
  \centering
  	\includegraphics[width=0.6\textwidth]{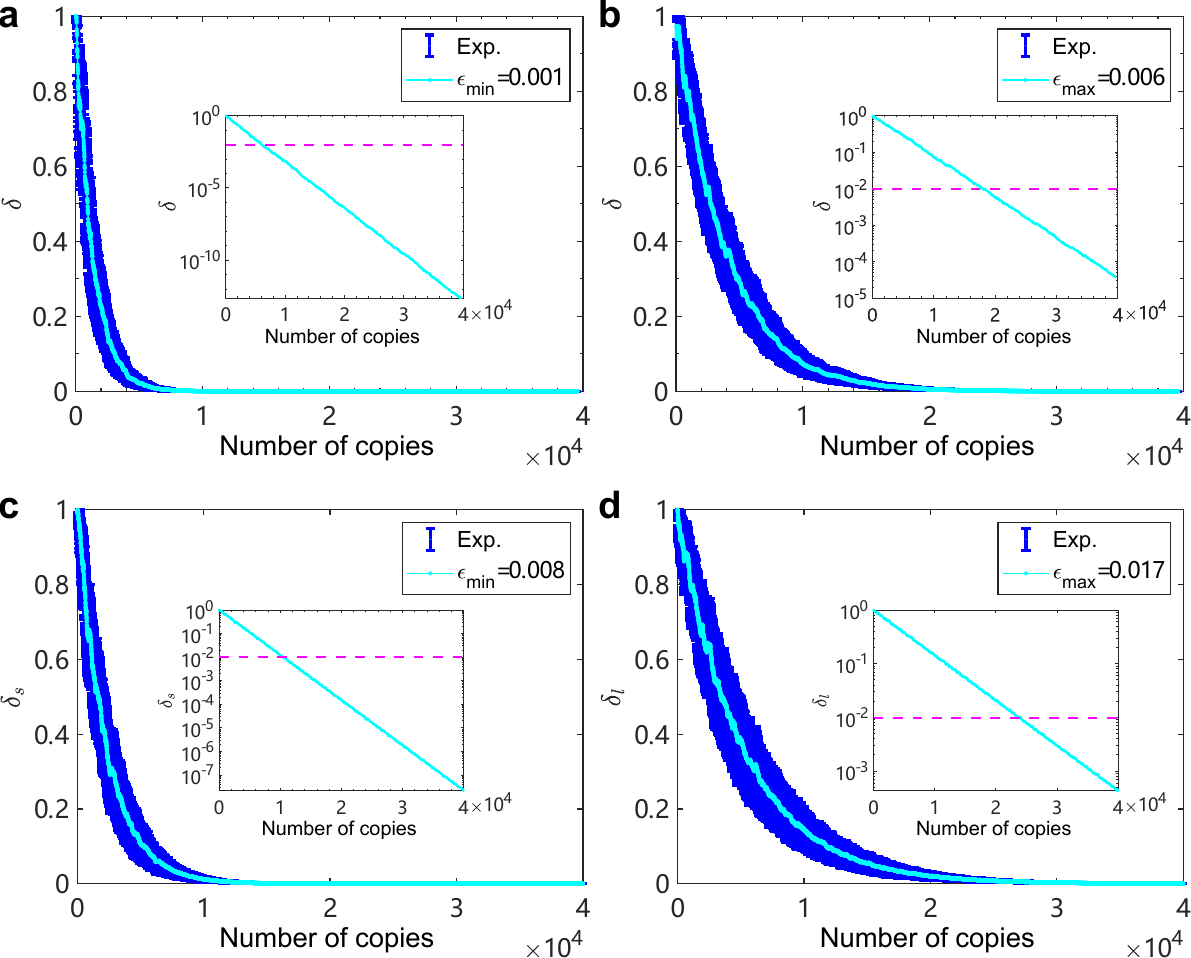}
  \caption{Experimental results for the verification of Task B. \textbf{a,b}~Nonadaptive strategy. The confidence parameter $\delta$ decreases with the increase of number of copies. After about 6000 copies, $\delta$ goes below 0.01 for \textbf{Case 2} (see inset of \textbf{a}). For \textbf{Case 1} (see inset of \textbf{b}), it takes about 17905 copies to reduce $\delta$ below 0.01. \textbf{c,d}~Adaptive strategy. The number of copies required to reduce $\delta_s$ and $\delta_l$ to be 0.01 for the two cases are about 10429 and 23645, respectively. Generally, it takes less number of copies for verifying \textbf{Case 2} because more space are allowed for the states to be found in the 0-$\mu N$ region. The blue is the experimental error bar (Exp.), which is obtained by 100 rounds of measurements for each coincidence. The insets show the log-scale plots, which indicates $\delta$ can reach a value below 0.01 with about thousands to tens of thousands of copies.}
  \label{fig4:DeltaN}
\end{figure}
\textbf{Task B}. We emphasize that Task B is considered under the assumption that the quantum device is either in Case 1 or in Case 2 as described above. These two cases are complementary and the confidence to assert whether the device belongs to Case 1 or Case 2 can be obtained according to different  values of $m_{\text{pass}}$. We refer to the Supplementary Note 3 for detailed information on judging the quantum device for these two cases. For each case, we can reduce the parameter $\delta$ by increasing the number of copies of the quantum state. Thus, the confidence $\delta_\text{B}=1-\delta$ to judge the device belongs to \textbf{Case 1}/\textbf{Case 2} is obtained. For the nonadaptive strategy, the passing probability $m_{\text{pass}}/N$ can finally reach a stable value 0.9986$\pm$0.0002 after about 1000 number of copies (see Supplementary Note 6). This value is smaller than the desired passing probability $\mu$ when we choose the infidelity $\epsilon_\text{min}$ to be 0.001. In this situation, we conclude the state belongs to \textbf{Case 2}. Conversely, the stable value is larger than the desired passing probability $\mu$ when we choose the infidelity $\epsilon_\text{max}$ to be 0.006. In this situation, we conclude the state belongs to \textbf{Case 1}. In Fig.~4, we present the results for the verification of \textbf{Task B}. First, we show the the confidence parameter $\delta$ versus the number of copies for the nonadaptive strategy in Fig.~4a, b. With about 6000 copies of quantum state, the $\delta$ parameter reaches 0.01 for \textbf{Case 2}. This indicates that the device belongs to \textbf{Case 1} with probability at most 0.01. In other words, there are at least 99\% confidence that we can say the device is in `bad' case after about 6000 measurements. Generally, more copies of quantum states are required to reach a same level $\delta$=0.01 for \textbf{Case 1}, because there are fewer portion for the number of passing events $m_{\text{pass}}$ to be chosen in the range of ${\mu}N$ to $N$. From Fig.~4b, we can see that it takes about 17905 copies of quantum state in order to reduce the parameter $\delta$ to be below 0.01. At this stage, we can say that the device belongs to \textbf{Case 2} with probability at most 0.01. That is, there are at least 99\% confidence that we can say the device is in `good' case after about 17905 measurements.

Figure~4c, d are the results of adaptive strategy. For the adaptive strategy, the passing probability $m_{\text{pass}}/N$ finally reaches a stable value 0.9914$\pm$0.0005 (see Supplementary Note 6), which is smaller than the nonadaptive measurement due to the limited fidelity of the EOMs' modulation. Correspondingly, the infidelity parameter for the two cases are chosen to be $\epsilon_\text{min}=0.008$ and $\epsilon_\text{max}=0.017$, respectively. We can see from the figure that it takes about 10429 number of copies for $\delta_s$ to be decreased to 0.01 when choosing $\epsilon_\text{min}$, which indicates that the device belongs to \textbf{Case 2} with at least 99\% confidence after about 10429 measurements. On the other hand, about 23645 number of copies are needed for $\delta_l$ to be decreased to 0.01 when choosing $\epsilon_\text{max}$, which indicates that the device belongs to \textbf{Case 1} with at least 99\% confidence after about 23645 measurements. Note that the difference of adaptive and nonadaptive comes from the different descent speed of $\delta$ versus the number of copies $N$, which results from the differences in passing probabilities and the infidelity parameters. See Supplementary Note 6 for detailed explanations.

From another perspective, we can fix $\delta$ and see how the parameter $\epsilon$ changes when increasing the number of copies. Figure~5 presents the variation of $\epsilon$ versus the number of copies in the log-log scale when we set the $\delta$ to be 0.10. At small number of copies, the infidelity is large and drops fast to a low level when the number of copies increases to be approximately 100. The decline becomes slow when the number of copies exceeds 100. It should be noted that the $\epsilon$ asymptotically tends to a value of 0.0036 (calculated by $1-\Delta_\epsilon=0.9986$) and 0.012 (calculated by $1-\Delta_\epsilon=0.9914$) for the nonadaptive and adaptive strategies, respectively. Therefore, we are still in the region of $m_{\text{pass}}/N{\geq}\mu$. We can also see that the scaling of $\epsilon$ versus $N$ is linear in the small number of copies region. We fit the data in the linear region with $\epsilon{\sim}N^{r}$ and obtain a slope $r${}$\sim$-0.88$\pm$0.03 for nonadaptive strategy and $r${}$\sim$-0.78$\pm$0.07 for adaptive strategy. This scaling exceeds the standard quantum limit $\epsilon{\sim}N^{-0.5}$ scaling~\cite{zhang2019experimental,giovannetti2011advances} for physical parameter estimation. Thus, our method is better for estimating the infidelity parameter $\epsilon$ than the classical metrology. Note that $m_{\text{pass}}/N$ is a good estimation for our state fidelity. If the state fidelity increases, the slope of linear region will decreases to the Heisenberg limit $\epsilon{\sim}N^{-1}$ in quantum metrology (see Supplementary Note 6).
\begin{figure}[!htbp]
  \centering
  	\includegraphics[width=0.48\textwidth]{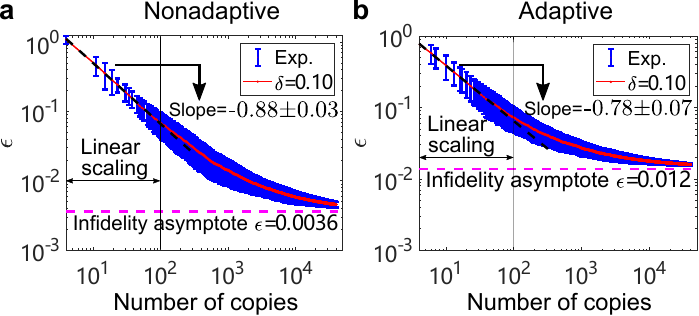}
  \caption{The variation of infidelity parameter versus the number of copies. \textbf{a}~Nonadaptive strategy and \textbf{b} Adaptive strategy. Here the data is plotted on a log-log scale. The confidence parameter $\delta$ is chosen to be 0.10. The parameter $\epsilon$ fast decays to a low value which is asymptotically close to the infidelity 0.0036 (Nonadaptive) and 0.012 (Adaptive) of the generated quantum state when increasing the number of copies. The fitting slopes for the linear scaling region are -0.88$\pm$0.03 and -0.78$\pm$0.07 for the nonadaptive and adaptive, respectively. The blue symbol is the experimental data with error bar (Exp.), which is obtained by 100 rounds of measurements for each coincidence.}
  \label{fig5:EpsilonN}
\end{figure}

\vspace*{0.5cm}
\noindent\sectionuser{Comparison with standard quantum state tomography}\quad
The advantage of the optimal verification strategy lies in that it requires fewer number of measurement settings, and more importantly, the number of copies to estimate the quantum states generated by a quantum device. In standard quantum state tomography,~\cite{Altepeter2005105} the minimum number of settings required for a complete reconstruction of the density matrix is $3^n$, where $n$ is the number of qubits. For two-qubit system, the standard tomography will cost nine settings whereas the present verification strategy only needs four and three measurement settings for the nonadaptive and adaptive strategies, respectively. To quantitatively compare the verification strategy with the standard tomography, we show the scaling of the parameters $\delta$ and $\epsilon$ versus the number of copies $N$ in Fig.~6. For each number of copies, the fidelity estimation $F\pm\Delta F$ can be obtained by the standard quantum state tomography. The $\delta$ of standard tomography is calculated by the confidence assuming normal distribution of the fidelity with mean $F$ and standard deviation $\Delta F$. The $\epsilon$ of standard tomography is calculated by $\epsilon=1-F$. The result of verification strategy is taken from the data in Figs.~4 and~5 for the nonadaptive strategy. For $\delta$ versus $N$, we fit the curve with equation $\delta=e^{g\cdot N}$, where $g$ is the scaling of $\log(\delta)$ with $N$. We obtain $g_\text{tomo}=-6.84\times 10^{-5}$ for the standard tomography and $g_\text{verif}=-7.35\times 10^{-4}$ for the verification strategy. This indicates that present verification strategy achieves better confidence than standard quantum state tomography given the same number of copies. For $\epsilon$ versus $N$, as shown in Fig.~6b, the standard tomography will finally reach a saturation value when increasing the number of copies. With the same number of copies $N$, the verification strategy obtains a smaller $\epsilon$, which indicates that the verification strategy can give a better estimation for the state fidelity than the standard quantum state tomography when small number of quantum states are available for a quantum device.
\begin{figure}[!htbp]
  \centering
  	\includegraphics[width=0.6\textwidth]{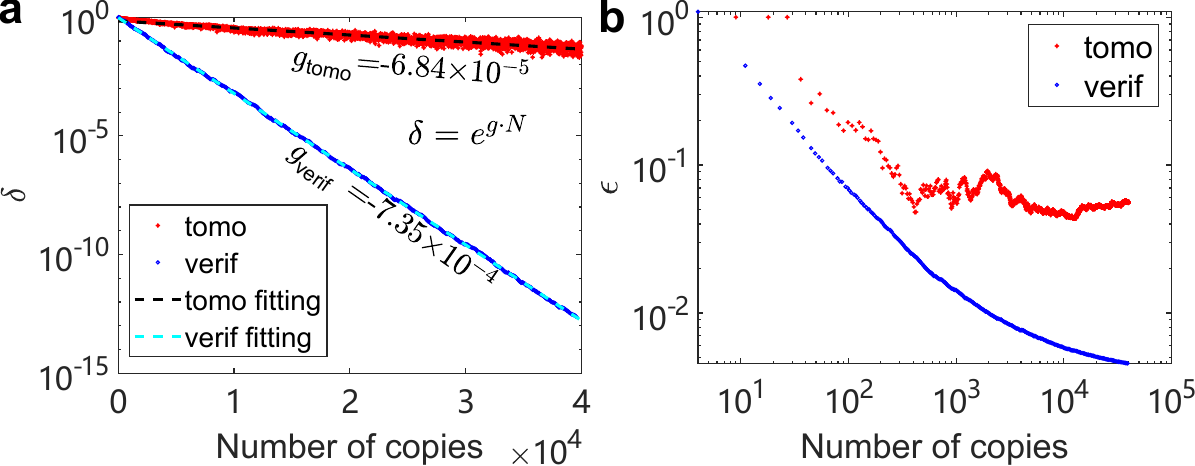}
  \caption{Comparison of standard quantum state tomography and present verification strategy. In the figure, we give the variation of \textbf{a}~$\delta$ and \textbf{b}~$\epsilon$ versus the number of copies $N$ by using standard quantum state tomography (tomo) and present verification strategy (verif). For standard tomography, the fidelity $F\pm\Delta F$ is first obtained from the reconstructed density matrix of each copy $N$. Then confidence parameter $\delta$ is estimated by assuming normal distribution of the fidelity with mean $F$ and standard deviation $\Delta F$. The infidelity parameter $\epsilon$ is estimated by $\epsilon=1-F$. Note that the experimental data symbols shown in \textbf{a} looks like lines due to the dense data points.}
  \label{fig6:Scaling}
\end{figure}

\vspace*{0.5cm}
\noindent{\secfont{DISCUSSION}}
\vspace*{0.1cm}

\noindent Our work, including experiment, data processing and analysis framework, can be used as a standardized procedure for verifying quantum states. In Task A, we give an estimation of the infidelity parameter $\epsilon_\text{exp}$ of the generated states and the confidence $\delta_\text{A}$ to produce the target quantum state by detecting certain number of copies. With the $\epsilon_\text{exp}$ obtained from Task A, we can choose $\epsilon_\text{max}$ or $\epsilon_\text{min}$ which divides our device to be Case 1 or Case 2. Task B is performed based on the chosen $\epsilon_\text{min}$ and $\epsilon_\text{max}$. We can have an estimation for the scaling of the confidence parameter $\delta$ versus the number of copies $N$ based on the analysis method of Task B. With a chosen $\delta$, we can also have an estimation for the scaling of the infidelity parameter $\epsilon$ versus $N$. With these steps, we can have a comprehensive judgement about how well our device really works.

In summary, we report experimental demonstrations for the optimal two-qubit pure state verification strategy with and without adaptive measurements. We give a clear discrimination and comprehensive analysis for the quantum states generated by a quantum device. Two tasks are proposed for practical applications of the verification strategy. The variation of confidence and infidelity parameter with the number of copies for the generated quantum states are presented. The obtained experimental results are in good agreement with the theoretical predictions. Furthermore, our experimental framework offers a precise estimation on the reliability and stability of quantum devices. This ability enables our framework to serve as a standard tool for analysing quantum devices. Our experimental framework can also be extended to other platforms.

\vspace*{1.0cm}
\noindent{\secfont{DATA AVAILABILITY}}
\vspace*{0.2cm}

\noindent The data that support the plots within this paper and other findings of this study are available from the corresponding author upon reasonable request.

\vspace*{1.0cm}
\noindent{\secfont{CODE AVAILABILITY}}
\vspace*{0.2cm}

\noindent The codes that support the plots within this paper and other findings of this study are available from the corresponding author upon reasonable request.

\vspace*{1.0cm}
\noindent{\secfont{ACKNOWLEDGEMENTS}}
\vspace*{0.2cm}

\noindent The authors thank B. Daki\'c for the helpful discussions. This work was supported by the National Key Research and Development Program of China (Nos. 2017YFA0303704 and 2019YFA0308704), the National Natural Science Foundation of China (Nos. 11674170 and 11690032), NSFC-BRICS (No. 61961146001), the Natural Science Foundation of Jiangsu Province (No. BK20170010), the Leading-edge technology Program of Jiangsu Natural Science Foundation (No. BK20192001), the program for Innovative Talents and Entrepreneur in Jiangsu, and the Fundamental Research Funds for the Central Universities.

\vspace*{1.0cm}
\noindent{\secfont{COMPETING INTERESTS}}
\vspace*{0.2cm}

\noindent A patent application related to this work is filed by Nanjing University on 29 May, 2020 in China. The application number is 202010475173.4 (Patent in China). The status of the application is now under patent pending.

\vspace*{1.0cm}
\noindent{\secfont{AUTHOR CONTRIBUTIONS}}
\vspace*{0.2cm}

\noindent X.-H.~J., K.-Y.~Q., Z.-Z.~C., Z.-Y.~C.~and X.-S.~M.~designed and performed the experiment. K.~W.~performed the theoretical analysis. X.-H.~J., K.-Y.~Q.~analysed the data. X.-H.~J., K.~W.~and X.-S.~M.~wrote the paper with input from all authors. All authors discussed the results and read the manuscript. F.-M.~S., S.-N.~Z.~and X.-S.~M.~supervised the work. X.-H.~J., K.~W., K.-Y.~Q., Z.-Z.~C. contributed equally to this work.

\vspace*{1.0cm}
\noindent{\secfont{REFERENCES}}
\clearpage
\setcounter{equation}{0}
\setcounter{figure}{0}
\setcounter{table}{0}
\makeatletter
\renewcommand{\figurename}{{\bfseries Supplementary Figure}}
\renewcommand{\thefigure}{{\bfseries \arabic{figure}}}
\renewcommand{\tablename}{{\bfseries Supplementary Table}}
\renewcommand{\thetable}{{\bfseries \arabic{table}}}
\renewcommand{\thesection}{Supplementary Note \arabic{section}}
\makeatother

\setcounter{page}{1}
\begin{center}
    \textbf{\LARGE Supplementary Information \---
Towards the standardization of quantum state verification using optimal strategies}
\end{center}
\bigskip

\section{The theory of optimal state verification}\label{sec1}

\noindent First we argue that it suffices to derive optimal verification strategies for the state $\ket{\Psi}$ defined in Eq.~(1) of
the main text. This is because locally unitarily equivalent pure states have locally unitarily equivalent optimal
verification strategies. This fact is proved in ref.~\cite[Lemma 2]{SPhysRevLett.120.170502} and we restate here for
completeness.
\begin{lemma}[Lemma 2 in the Supplemental Material of ref.~\cite{SPhysRevLett.120.170502}]\label{lemma:2}
Given any two-qubit state $\ket{\psi}$ with optimal strategy $\Omega$, a locally unitarily equivalent state $(U\otimes
V)\ket{\psi}$, where $U$ and $V$ are unitaries, has optimal strategy $(U\otimes V)\Omega(U\otimes V)^\dagger$.
\end{lemma}

\subsection{Optimal nonadaptive strategy}\label{sec:nonadaptive}

In this section, we briefly summarize the optimal nonadaptive strategy.
According to ref.~\cite{SPhysRevLett.120.170502}, any optimal strategy for verifying state the state $\ket{\Psi}$ ($\theta\in(0,\pi/4)$) defined in Eq.~(1) of the main text, that accepts $\ket{\Psi}$ with
certainty and satisfies the properties of locality, projectivity, and trusty, can be expressed as a strategy involving
the following four measurements,
\begin{align}\label{eq:nonTheo}
\Omega_{\opn{opt}} 
= \alpha(\theta)P_{ZZ}^+ + \frac{1-\alpha(\theta)}{3} \sum_{k=1}^3\left[\1 - \proj{u_k}\otimes\proj{v_k}\right],
\end{align}
where $\alpha(\theta) = (2 - \sin(2\theta))/(4 + \sin(2\theta))$, $P_{ZZ}^+ = \proj{HH} + \proj{VV}$ is the projector
onto the positive eigenspace of the tensor product of Pauli matrix $Z\otimes Z$, and the states $\{\ket{u_k}\}$ and
$\{\ket{v_k}\}$ are written explicitly in the following:
\begin{align}
\ket{u_1} &= \frac{1}{\sqrt{1+\tan\theta}}\ket{H} + \frac{e^{\frac{2\pi i}{3}}}{\sqrt{1+\cot\theta}}\ket{V}, \\
\ket{v_1} &= \frac{1}{\sqrt{1+\tan\theta}}\ket{H} + \frac{e^{\frac{\pi i}{3}}}{\sqrt{1+\cot\theta}}\ket{V},\\
\ket{u_2} &= \frac{1}{\sqrt{1+\tan\theta}}\ket{H} + \frac{e^{\frac{4\pi i}{3}}}{\sqrt{1+\cot\theta}}\ket{V}, \\
\ket{v_2} &= \frac{1}{\sqrt{1+\tan\theta}}\ket{H} + \frac{e^{\frac{5\pi i}{3}}}{\sqrt{1+\cot\theta}}\ket{V},\\
\ket{u_3} &= \frac{1}{\sqrt{1+\tan\theta}}\ket{H} + \frac{1}{\sqrt{1+\cot\theta}}\ket{V}, \\
\ket{v_3} &= \frac{1}{\sqrt{1+\tan\theta}}\ket{H} - \frac{1}{\sqrt{1+\cot\theta}}\ket{V}.
\end{align}
Correspondingly, the second largest eigenvalue of $\Omega_{\opn{opt}}$ is given by
\begin{align}\label{eq:Omega-2-non}
    \lambda_2(\Omega_{\opn{opt}}) = \frac{2 + \sin(2\theta)}{4 + \sin(2\theta)}. 
\end{align}
Substitute this value into Eq.~(3) of the main text, we find the optimal number of measurements required to verify
$\ket{\Psi}$ with infidelity $\epsilon$ and confidence $1-\delta$ satisfies
\begin{equation}
n_{\opn{opt}} 
\approx \frac{1}{[1-\lambda_2(\Omega_{\opn{opt}})]\epsilon}\ln\frac{1}{\delta}
= \left(2+\sin\theta\cos\theta\right)\frac{1}{\epsilon}\ln\frac{1}{\delta}.
\end{equation}

For analysis, we consider the spectral decomposition of $\Omega_{\opn{opt}}$. This decomposition is helpful for
computing the passing probability of the states being verified. Let $\ket{\Psi^\perp} :=
\cos\theta\ket{HH}-\sin\theta\ket{VV}$. One can check that $\{\ket{\Psi},\ket{\Psi^\perp},\ket{HV},\ket{VH}\}$ forms
an orthonormal basis of a two-qubit space. The spectral decomposition of $\Omega_{\opn{opt}}$ is
\begin{align}\label{eq:Omega-non-sd}
  \Omega_{\opn{opt}} 
= \ket{\Psi}\bra{\Psi}
+ \lambda_2(\Omega_{\opn{opt}})
  \left(\proj{\Psi^\perp}+\proj{VH}+\proj{HV}\right),
\end{align}
where $\lambda_2(\Omega_{\opn{opt}})$ is given in Eq.~\eqref{eq:Omega-2-non}.

\subsection{Optimal adaptive strategy using one-way classical communication}\label{sec:adaptive}

In this section, we briefly summarize the optimal adaptive strategy using one-way classical communication. According to
ref.~\cite{Swang2019optimal} any optimal strategy for verifying the state $\ket{\Psi}$ ($\theta\in(0,\pi/4)$)
defined in Eq.~(1) of the main text, that accepts $\ket{\Psi}$ and can be implemented by local measurements together
with one-way classical communication, can be expressed as a strategy involving the following three measurements,
\begin{align}\label{eq:adaTheo}
\Omega^\rightarrow_{\opn{opt}} 
= \beta(\theta)P_{ZZ}^+ + \frac{1-\beta(\theta)}{2}T_1 + \frac{1-\beta(\theta)}{2}T_2,
\end{align}
where $\beta(\theta)=\cos^2\theta/(1+\cos^2\theta)$ and
\begin{align}
T_1 &= \proj{+}\otimes\proj{v_{+}} + \proj{-}\otimes\proj{v_{-}},  \\
T_2 &= \proj{R}\otimes\proj{w_{+}} + \proj{L}\otimes\proj{w_{-}}, 
\end{align} 
such that
\begin{alignat}{2}
  \ket{+} &= \frac{\ket{V}+\ket{H}}{\sqrt{2}}, &\quad 
  \ket{-} &= \frac{\ket{V}-\ket{H}}{\sqrt{2}},  \\
  \ket{v_+} &= \cos\theta\ket{V}+\sin\theta\ket{H},&\quad 
  \ket{v_-} &= \cos\theta\ket{V}-\sin\theta\ket{H},  \\
  \ket{R} &= \frac{\ket{V}+i\ket{H}}{\sqrt{2}}, &\quad
  \ket{L} &= \frac{\ket{V}-i\ket{H}}{\sqrt{2}},  \\
  \ket{w_+} &= \cos\theta\ket{V}-i\sin\theta\ket{H}, &\quad 
  \ket{w_-} &= \cos\theta\ket{V}+i\sin\theta\ket{H}.
\end{alignat}
Correspondingly, the second largest eigenvalue of $\Omega^\rightarrow_{\opn{opt}}$ is given by
\begin{align}\label{eq:Omega-2-adp}
    \lambda_2(\Omega^\rightarrow_{\opn{opt}}) = \frac{\cos^2\theta}{1+\cos^2\theta}. 
\end{align}
Substitute this value into Eq.~(3) of the main text, we find the optimal number of measurements required to verify
$\ket{\Psi}$ with infidelity $\epsilon$ and confidence $1-\delta$ satisfies
\begin{equation}
n^\rightarrow_{\opn{opt}} 
\approx \frac{1}{[1-\lambda_2(\Omega^\rightarrow_{\opn{opt}})]\epsilon}\ln\frac{1}{\delta}
= (1+\cos^2\theta)\frac{1}{\epsilon}\ln\frac{1}{\delta}.
\end{equation}
The spectral decomposition of $\Omega^\rightarrow_{\opn{opt}}$ is given by
\begin{align}\label{eq:Omega-adp-sd}
  \Omega^\rightarrow_{\opn{opt}} 
= \ket{\Psi}\bra{\Psi}
+ \lambda_2(\Omega^\rightarrow_{\opn{opt}})\left(\proj{\Psi^\perp} + \proj{HV}\right)
+ \lambda_4(\Omega^\rightarrow_{\opn{opt}})\proj{VH},
\end{align}
where $\lambda_2(\Omega^\rightarrow_{\opn{opt}})$ is the second largest eigenvalue given in Eq.~\eqref{eq:Omega-2-adp} and
$\lambda_4(\Omega^\rightarrow_{\opn{opt}})=\sin^2\theta/(1+\cos^2\theta)$ is the fourth (which is also the smallest)
eigenvalue of $\Omega^\rightarrow_{\opn{opt}}$.

\subsection{Optimal verification of Bell state}

As pointed out in ref.~\cite{SPhysRevLett.120.170502}, the above nonadaptive strategy is actually only optimal for
$\theta\in(0,\pi/4)\cup(\pi/4,\pi/2)$. That is to say, it is no longer optimal for verifying the Bell state, for which
$\theta=\pi/4$. In this section, we summary explicitly the optimal strategy for verifying the Bell state of the
following form:
\begin{equation}\label{eq:Bell}
  \ket{\Phi^+} = \frac{1}{\sqrt{2}}\left(\ket{HH} + \ket{VV} \right).
\end{equation}
The first proposal for testing the Bell state $\ket{\Phi^+}$ is given in ref.~\cite{SHayashiJPAMG2006}. According to Section 7.1 of ref.~\cite{SHayashiJPAMG2006} and Eq.~(8) of ref.~\cite{SPhysRevLett.120.170502}, the optimal strategy for verifying $\ket{\Phi^+}$ implementable by locally projective measurements, can be expressed as a strategy involving the following three measurements,
\begin{align}
\Omega_{\opn{opt,Bell}} = \frac{1}{3}\left(P_{XX}^+ + P_{YY}^- + P_{ZZ}^+\right),
\end{align}
where $P_{XX}^+$ is the projector onto the positive eigenspace of the tensor product of Pauli matrix $X\otimes X$ and
$P_{YY}^-$ is the projector onto the negative eigenspace of the tensor product of Pauli matrix $Y\otimes Y$, similarly
defined as that of $P_{ZZ}^+$.
We can check that $\lambda_2(\Omega_{\opn{opt,Bell}})=1/3$
and thus  the optimal number of measurements required to verify
$\ket{\Phi^+}$ with infidelity $\epsilon$ and confidence $1-\delta$ satisfies
\begin{equation}
n_{\opn{opt,Bell}} 
\approx \frac{1}{[1-\lambda_2(\Omega_{\opn{opt,Bell}})]\epsilon}\ln\frac{1}{\delta}
= \frac{3}{2\epsilon}\ln\frac{1}{\delta}.
\end{equation}

\subsection{Optimal verification of a product state}

For the product state $\ket{HV}$ ($\theta=0$ or $\pi/2$), the optimal strategy is provably given by $\Omega_{\opn{opt,pd}} = \proj{HV}$.
Obviously, $\lambda_2(\Omega_{\opn{opt,pd}})=1$ and thus the optimal number of measurements required to verify
$\ket{HV}$ with infidelity $\epsilon$ and confidence $1-\delta$ satisfies
\begin{equation}
n_{\opn{opt,pd}} 
\approx \frac{1}{[1-\lambda_2(\Omega_{\opn{opt,pd}})]\epsilon}\ln\frac{1}{\delta}
= \frac{1}{\epsilon}\ln\frac{1}{\delta}.
\end{equation}


\section{The experiment of optimal state verification}\label{sec4}

\subsection{The verification strategies for experimentally generated state}

\noindent In our experiment setting, the target state has the form
\begin{align}\label{target}
  \ket{\psi(\theta,\phi)} = \sin\theta\ket{HV} + e^{i\phi}\cos\theta\ket{VH},
\end{align}
where $\theta\in[0,\pi/4]$ and $\phi\in[0,2\pi]$, which will be determined by the experimental data. We can see that $\ket{\psi(\theta,\phi)}$ is equivalent to $\ket{\Psi}$ defined in Eq.~(1) of the
main text through the following unitary operator:
\begin{align}
    \mathbb{U} \equiv (\1\otimes\kappa) = \begin{pmatrix} 1 & 0 \\ 0 & 1\end{pmatrix}\otimes
     \begin{pmatrix} 0 & e^{i\phi} \\ 1 & 0 \end{pmatrix}.
\end{align}
This is indeed the case since 
\begin{align}
(\1\otimes\kappa)\ket{\Psi}
= \begin{pmatrix} 0 & e^{i\phi} & 0 & 0 \\ 
                  1 & 0 & 0 & 0 \\
                  0 & 0 & 0 & e^{i\phi} \\
                  0 & 0 & 1 & 0\end{pmatrix}
  \begin{pmatrix} \sin\theta \\ 0 \\ 0 \\ \cos\theta \end{pmatrix}
= \begin{pmatrix} 0 \\ \sin\theta \\ e^{i\phi}\cos\theta \\ 0 \end{pmatrix}
= \ket{\psi(\theta,\phi)}. 
\end{align} 
Then Lemma~\ref{lemma:2} together with the optimal strategies for verifying $\ket{\Psi}$ summarized in the
last section yields the optimal strategies for verifying $\ket{\psi(\theta,\phi)}$.
For completeness, we write down explicitly the measurements of these strategies.

\vspace*{0.1in}
\textbf{Optimal nonadaptive strategy for $\ket{\psi(\theta,\phi)}$.} By Lemma~\ref{lemma:2} and Eq.~\eqref{eq:nonTheo}, the optimal nonadaptive
strategy for verifying $\ket{\psi(\theta,\phi)}$ has the following form,
\begin{align}
\Omega_{\opn{opt}}' = \mathbb{U}\Omega_{\opn{opt}}\mathbb{U}^\dagger
= \alpha(\theta)P_0 + \frac{1-\alpha(\theta)}{3}(P_1 + P_2 + P_3),
\end{align}
where
\begin{align}\label{eq:P0}
  P_0 = \mathbb{U} P_{ZZ}^+ \mathbb{U}^\dagger = \proj{H}\otimes\proj{V} + \proj{V}\otimes\proj{H},
\end{align}
and $P_i$ for $i=1,2,3$ satisfies $P_i=\proj{\widetilde{u}_i}\otimes\proj{\widetilde{v}_i}$ such that
\begin{align}
    \ket{\widetilde{u}_1}
&=  \ket{u_1} 
 =  \frac{1}{\sqrt{1+\tan\theta}}\ket{H}+\frac{e^{\frac{2\pi i}{3}}}{\sqrt{1+\cot\theta}}\ket{V}, \label{eq:P11}\\
\ket{\widetilde{v}_1}
&=  \kappa\ket{v_1} 
 =  \frac{1}{\sqrt{1+\tan\theta}}\ket{V}+\frac{e^{\frac{ \pi i}{3}}e^{i\phi}}{\sqrt{1+\cot\theta}}\ket{H}, \label{eq:P12}\\
    \ket{\widetilde{u}_2}
&=  \ket{u_2} 
 =  \frac{1}{\sqrt{1+\tan\theta}}\ket{H}+\frac{e^{\frac{4\pi i}{3}}}{\sqrt{1+\cot\theta}}\ket{V}, \label{eq:P21}\\
    \ket{\widetilde{v}_2} 
&=  \kappa\ket{v_2} 
 =  \frac{1}{\sqrt{1+\tan\theta}}\ket{V}+\frac{e^{\frac{5\pi i}{3}}e^{i\phi}}{\sqrt{1+\cot\theta}}\ket{H}, \label{eq:P22}\\
    \ket{\widetilde{u}_3}
&=  \ket{u_3} 
 =  \frac{1}{\sqrt{1+\tan\theta}}\ket{H}+\frac{1}{\sqrt{1+\cot\theta}}\ket{V}, \label{eq:P31}\\
\ket{\widetilde{v}_3} 
&= \kappa\ket{v_3} 
 = \frac{1}{\sqrt{1+\tan\theta}}\ket{V}+\frac{e^{\frac{3\pi i}{3}}e^{i\phi}}{\sqrt{1+\cot\theta}}\ket{H}. \label{eq:P32}
\end{align}

\vspace*{0.1in}
\textbf{Optimal adaptive strategy using one-way classical communication for $\ket{\psi(\theta,\phi)}$.} By
Lemma~\ref{lemma:2} and Eq.~\eqref{eq:adaTheo}, the optimal adaptive strategy for verifying $\ket{\psi(\theta,\phi)}$, when one-way classical
communication is allowed, has the following form,
\begin{align}
\Omega'^\rightarrow_{\opn{opt}} 
= \beta(\theta)\widetilde{T}_0 
+ \frac{1-\beta(\theta)}{2}\widetilde{T}_1 
+ \frac{1-\beta(\theta)}{2}\widetilde{T}_2,
\end{align}
where
\begin{align}
\widetilde{T}_0 
&= \mathbb{U}P_{ZZ}^+\mathbb{U}^\dagger
 = \proj{H}\otimes\proj{V} + \proj{V}\otimes\proj{H}, \label{eqAdp1} \\
\widetilde{T}_1 
&= \mathbb{U}T_1\mathbb{U}^\dagger \nonumber\\
&= \proj{+}\otimes\kappa\proj{v_{+}}\kappa^\dagger + \proj{-}\otimes\kappa\proj{v_{-}}\kappa^\dagger \nonumber\\
&\equiv \proj{+}\otimes\proj{\widetilde{v}_{+}} + \proj{-}\otimes\proj{\widetilde{v}_{-}}, \label{eqAdp2}\\
\widetilde{T}_2 
&= \mathbb{U}T_2\mathbb{U}^\dagger \nonumber\\
&= \proj{R}\otimes\kappa\proj{w_{+}}\kappa^\dagger + \proj{L}\otimes\kappa\proj{w_{-}}\kappa^\dagger \nonumber\\
&\equiv  \proj{R}\otimes\proj{\widetilde{w}_{+}} + \proj{L}\otimes\proj{\widetilde{w}_{-}}, \label{eqAdp3}
\end{align} 
and
\begin{alignat}{2}
  \ket{+}
&= \frac{\ket{V} + \ket{H}}{\sqrt{2}}, &\quad
  \ket{-}
&= \frac{\ket{V} - \ket{H}}{\sqrt{2}},\\
  \ket{R}
&= \frac{\ket{V} + i\ket{H}}{\sqrt{2}}, &\quad
  \ket{L}
&= \frac{\ket{V} - i\ket{H}}{\sqrt{2}},\\
  \ket{\widetilde{v}_+} 
&= \kappa\ket{v_+} = e^{i\phi}\cos\theta\ket{H}+\sin\theta\ket{V},&\quad 
  \ket{\widetilde{v}_-} 
&= \kappa\ket{v_-} = e^{i\phi}\cos\theta\ket{H}-\sin\theta\ket{V},  \label{eq:vw1}\\
  \ket{\widetilde{w}_+} 
&= \kappa\ket{w_+} = e^{i\phi}\cos\theta\ket{H}-i\sin\theta\ket{V},&\quad
  \ket{\widetilde{w}_-} 
&= \kappa\ket{w_-} = e^{i\phi}\cos\theta\ket{H}+i\sin\theta\ket{V}. \label{eq:vw2}
\end{alignat}

\vspace*{0.1in}

\textbf{Experimental optimal verification of Bell state.} Experimentally, our target Bell state has the following form
\begin{equation}
  \ket{\Phi^-} = \frac{1}{\sqrt{2}}\left(\ket{HV} - \ket{VH} \right),
\end{equation}
which is locally unitarily equivalent to the Bell state $\ket{\Phi^+}$ defined in Eq.~\eqref{eq:Bell} through
\begin{equation}
  (\1\otimes XZ)\ket{\Psi^+}
= \frac{1}{\sqrt{2}}(\1\otimes XZ)\left(\ket{HH} + \ket{VV} \right)
= \frac{1}{\sqrt{2}}\left(\ket{HV} - \ket{VH} \right)
= \ket{\Phi^-}.
\end{equation}
By Lemma~\ref{lemma:2}, the optimal adaptive strategy for verifying $\ket{\Phi^-}$ is given by
\begin{align}
\Omega'_{\opn{opt,Bell}} = \frac{1}{3}\left(M_1 + M_2 + M_3\right),
\end{align}
where
\begin{align}
  M_1 &= (\1\otimes XZ) P_{XX}^+ (\1\otimes XZ)^\dagger
      = \proj{-}\otimes\proj{+} + \proj{+}\otimes\proj{-}, \\
  M_2 &= (\1\otimes XZ) P_{YY}^- (\1\otimes XZ)^\dagger 
      = \proj{L}\otimes\proj{R} + \proj{R}\otimes\proj{L}, \\
  M_3 &= (\1\otimes XZ) P_{ZZ}^+ (\1\otimes XZ)^\dagger
      = \proj{H}\otimes\proj{V} + \proj{V}\otimes\proj{H}.
\end{align}
In the above derivation, we have used the following relations:
\begin{align}
XZ\ket{H} &= \ket{V},\quad\hspace{0.37cm} XZ\ket{V} = -\ket{H}, \\
XZ\ket{+} &= -\ket{+},\quad\hspace{0.08cm} XZ\ket{-} = \ket{+}, \\
XZ\ket{L} &= -i\ket{L},\quad XZ\ket{R} = i\ket{R}.
\end{align}

\subsection{The experimental verification procedure}

We now describe the procedure on how to verify the states $\ket{\psi(\theta,\phi)}$ described above. The parameter
$\theta$ and $\phi$ can be chosen to generate arbitrary pure states according to the angles of our wave plate (QWP1 and
HWP1 in Fig.~2 of the main text). In order to tune the angle $\theta$ to generate several states between 0 and $\pi/2$,
we let $\theta=k\cdot\pi/10$ ($k=1,2,3,4$). For reference, we list the tomographic parameters for $k=2$ (k2), maximally
entangled state (Max) and product state (HV) used in this article in Supplementary Table~\ref{truestate}.
\begin{table}[!htbp]
	\centering
	\caption{Target states verified in our experiment.}
	\label{truestate}
	\begin{tabular}{|c|c|c|c|c|c|}
		\hline
		k       & $k\pi/10$ & $\theta$ & $\phi$ & Fidelity            & Passing Prob. $m_{\text{pass}}/N$ \\
		\hline
		k2      & 0.6283    & 0.6419   & 3.2034 & 0.9964$\pm$0.0002   & 0.9986$\pm$0.0002              \\
		\hline
		Max     & $\pi/4$   & $\pi/4$  & $\pi$  & 0.9973$\pm$0.0002   & 0.9982$\pm$0.0002              \\
		\hline
		HV      & $\pi/2$   & $\pi/2$  & 0      & 0.9992$\pm$0.0001   & 0.9992$\pm$0.0001              \\
		\hline
	\end{tabular}
\end{table}

\summary{Determine the target state.} For our target state defined in Eq.~\eqref{target}, $\theta$ and $\phi$ are two
parameters to be set experimentally. For each $k$, we first calculate the ratio of $HV$ component to $VH$ component,
i.e.,~$(\sin\theta/\cos\theta)^2$, in the target state. Based on this ratio, we rotate the angles of QWP1 and HWP1 of
pump light to generate this target state. Then, the density matrix $\rho$ of the target state can be obtained by means
of the maximum likelihood estimations. The overlap between $\ket{\psi}$ and $\rho$ is used as our objective function. A
global search algorithm is used to minimize the objective function $1-\bra{\psi}\rho\ket{\psi}$ with $\theta$ and
$\phi$ as two free parameters in $\ket{\psi}$. Finally, we obtain the values of these two parameters with our
optimization program~\footnote{During the optimization, we use the program offered by \href{http://research.physics.illinois.edu/QI/Photonics/Tomography/}{Kwiat Quantum Information Group}. The code can be downloaded from their website. Last accessed: 2018-05-10.}.

\summary{The projective measurements.}
The nonadaptive strategy requires four projective measurements. From the expression of the four projectors (Eq.~(\ref{eq:P0}-\ref{eq:P32})), we can see that $P_0$ is the standard $\sigma_z$ projection while $P_i$ ($i=1,2,3$) are three general projectors that are realized by rotating the angles of wave plates in the state analyser (see Fig.~2 in main text). The function of $P_i$ projectors is to transform the $\ket{\widetilde{u}_i}$ and $\ket{\widetilde{v}_i}$ states to the $\ket{H}$ state with combinations of QWP and HWP. Therefore, we treat the angles of QWP and HWP as two quantities to realize the transformation utilizing Jones matrix method. By solving the equations, we find the angles of wave plates that realizes the four projective measurements.

For the adaptive strategy, three measurements $\widetilde{T}_0$, $\widetilde{T}_1$ and $\widetilde{T}_2$ are required. From Eq.~(\ref{eqAdp1}-\ref{eq:vw2}), we can see that $\widetilde{T}_0$ is the standard $\sigma_z$ projection. The $\widetilde{T}_1$ and $\widetilde{T}_2$ measurements requires classical communication to transmit the results of Alice's measurements to Bob. Then Bob applies the corresponding measurements using the electro-optic modulators (EOMs) according to Alice's results. For $\widetilde{T}_1$, Alice implements the $\{\ket{+}, \ket{-}\}$ orthogonal measurements while Bob's two measurements $\ket{\widetilde{v}_{+}}$ and $\ket{\widetilde{v}_{-}}$ are nonorthogonal. For $\widetilde{T}_2$, Alice's measurements $\{\ket{R}, \ket{L}\}$ are orthogonal whereas Bob need to apply the nonorthogonal measurements $\ket{\widetilde{w}_{+}}$ and $\ket{\widetilde{w}_{-}}$ accordingly. In the following, we describe how to realize these adaptive measurements by means of local operations and classical communication (LOCC) in real time.

In order to implement the $\{\ket{\widetilde{v}_{+}}, \ket{\widetilde{v}_{-}}\}$ and $\{\ket{\widetilde{w}_{+}}, \ket{\widetilde{w}_{-}}\}$ measurements at Bob's site according to Alice's measurement outcomes, we use two EOMs, as shown in Supplementary Figure~\ref{fig:Adp}. The $\ket{\widetilde{v}_{+}}$ and $\ket{\widetilde{w}_{+}}$ polarized states are rotated to $H$ polarization with EOM1, which finally exit from the transmission port of polarized beam splitter (PBS, the wave plates before PBS are rotated to $HV$ bases). Correspondingly, the $\ket{\widetilde{v}_{-}}$ and $\ket{\widetilde{w}_{-}}$ are transformed to the $V$ polarization with EOM2 and exit from the reflection port of PBS. We use the output of $\ket{+}$/$\ket{R}$ to trigger the EOM1 and the output of $\ket{-}$/$\ket{L}$ to trigger the EOM2, respectively. This design guarantees that only one EOM takes effect and the other executes the identity operation for each of the generated state. In Supplementary Table~\ref{AdpMeas}, we list the operations realized with these two EOMs. We can see that the adaptive measurements are genuinely realized with our setup.
\begin{figure}[!htbp]
  \centering
  	\includegraphics[width=0.5\textwidth]{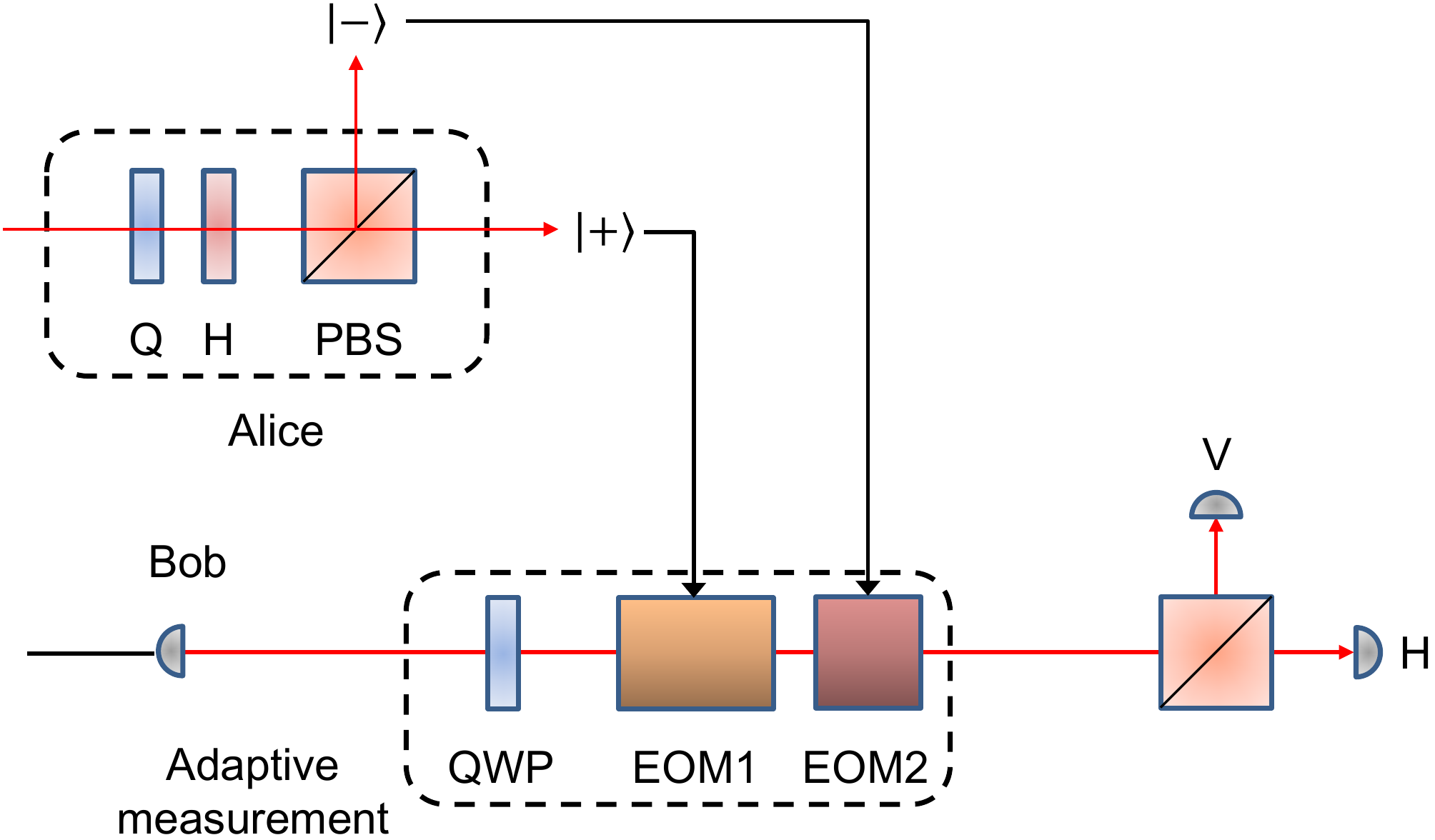}
  \caption{The adaptive measurements implemented with electro-optic modulator (EOM). Here we take $\widetilde{T}_1$ measurement as an illustration. The output signal of $\ket{+}$ is used to trigger EOM1 to execute the $\ket{\widetilde{v}_{+}}$ measurement, while the output signal of $\ket{-}$ is used to trigger EOM2 to execute the $\ket{\widetilde{v}_{-}}$ measurement. Only one EOM is triggered at a time and the other EOM is an identity operation, which guarantees the adaptive measurements are realized in real time. A QWP before the two EOMs is used to compensate the extra phase from $\widetilde{T}_1$ to $\widetilde{T}_2$ measurement.}
  \label{fig:Adp}
\end{figure}

There are two key ingredients for our adaptive measurement that distinguish it from the nonadaptive measurement. First, Alice produces binary outcomes for a given basis. In each run, only one of these two outcomes will be obtained. Bob's measurements have to be adaptive according to Alice's measurement results. Second, it doesn't require fast changes among the $\widetilde{T}_0$, $\widetilde{T}_1$ and $\widetilde{T}_2$ measurements. This can be seen from the fact that $\{\ket{\widetilde{w}_{+}}, \ket{\widetilde{w}_{-}}\}$ have an extra $\pi/2$ phase compared with the $\{\ket{\widetilde{v}_{+}}, \ket{\widetilde{v}_{-}}\}$ bases (see Eqs.~(\ref{eq:vw1},\ref{eq:vw2})). We add a QWP (see Supplementary Figure~\ref{fig:Adp}) before the two EOMs to compensate this extra phase. We calibrate the optical axis of these two EOMs and fix them. If we want to perform from $\widetilde{T}_1$ to $\widetilde{T}_2$ measurement or vice versa, we add or remove this QWP. For $\widetilde{T}_0$ measurement, we turn off the drivers of these two EOMs and remove the QWP. In the single-photon experiment, we can optimize the modulation contrast of these two EOMs with the coincidence ratio of $CC_{+\widetilde{v}_{+}}$/$CC_{+\widetilde{v}_{+}^\perp}$ and $CC_{-\widetilde{v}_{-}}$/$CC_{-\widetilde{v}_{-}^\perp}$ as our optimization goals (see Supplementary Table~\ref{AdpMeas}).
\begin{table}[!htbp]
	\centering
	\caption{Adaptive measurements \{$\widetilde{T}_0$, $\widetilde{T}_1$, $\widetilde{T}_2$\} realized with two EOMs.}\label{AdpMeas}
	\begin{threeparttable}
	\begin{tabular}{|c|c|c|c|c|c|c|}
		\hline
		Setting & Alice   & \multicolumn{2}{c|}{Bob's operations}                & Output & Modulation contrast* & Success probability      \\
		\hline
		\multirow{2}{*}{$\widetilde{T}_0$} & $H$     & $I$ (EOM1) & $I$ (EOM2) & --- & $CC_{HV}$/$CC_{HH}$ &  \multirow{2}{*}{$\frac{CC_{HV}+CC_{VH}}{CC_{HV}+CC_{HH}+CC_{VV}+CC_{VH}}$}  \\
		\cline{2-6}
		                                   & $V$     & $I$ (EOM1) & $I$ (EOM2) & --- & $CC_{VH}$/$CC_{VV}$ &   \\
		\hline
		\multirow{2}{*}{$\widetilde{T}_1$} & $+$     & $\widetilde{v}_{+}$ (EOM1) & $I$ (EOM2) & $\widetilde{v}_{+}\rightarrow H$ & $CC_{+\widetilde{v}_{+}}$/$CC_{+\widetilde{v}_{+}^\perp}$ &  \multirow{2}{*}{$\frac{CC_{+\widetilde{v}_{+}}+CC_{-\widetilde{v}_{-}}}{CC_{+\widetilde{v}_{+}}+CC_{+\widetilde{v}_{-}}+CC_{-\widetilde{v}_{+}}+CC_{-\widetilde{v}_{-}}}$}  \\
		\cline{2-6}
		                                   & $-$     & $I$ (EOM1) & $\widetilde{v}_{-}$ (EOM2) & $\widetilde{v}_{-}\rightarrow V$ & $CC_{-\widetilde{v}_{-}}$/$CC_{-\widetilde{v}_{-}^\perp}$ &   \\
		\hline
		\multirow{2}{*}{$\widetilde{T}_2$} & $R$     & $\widetilde{w}_{+}$ (EOM1) & $I$ (EOM2) & $\widetilde{w}_{+}\rightarrow H$ & $CC_{R\widetilde{w}_{+}}$/$CC_{R\widetilde{w}_{+}^\perp}$ &  \multirow{2}{*}{$\frac{CC_{R\widetilde{w}_{+}}+CC_{L\widetilde{w}_{-}}}{CC_{R\widetilde{w}_{+}}+CC_{R\widetilde{w}_{-}}+CC_{L\widetilde{w}_{+}}+CC_{L\widetilde{w}_{-}}}$}  \\
		\cline{2-6}
		                                   & $L$     & $I$ (EOM1) & $\widetilde{w}_{-}$ (EOM2) & $\widetilde{w}_{-}\rightarrow V$ & $CC_{L\widetilde{w}_{-}}$/$CC_{L\widetilde{w}_{-}^\perp}$ &   \\
		\hline
	\end{tabular}
	\begin{tablenotes}
	\footnotesize
	\item[*] $\widetilde{v}_{+}^\perp$/$\widetilde{w}_{+}^\perp$ and $\widetilde{v}_{-}^\perp$/$\widetilde{w}_{-}^\perp$ are the orthogonal states to $\widetilde{v}_{+}$/$\widetilde{w}_{+}$ and $\widetilde{v}_{-}$/$\widetilde{w}_{-}$, respectively.
	\end{tablenotes}
	\end{threeparttable}
\end{table}

\section{Applying the verification strategy to Task B}

\noindent In this section, we explain in detail how we use the verification strategy proposed in~\ref{sec1} to execute
\textbf{Task B}. Generally speaking, \textbf{Task B} is to distinguish whether the quantum device outputs states
$\epsilon$-far alway from the target state $\ket{\Psi}$ or not. Let $\cS$ be the set of states that are at least
$\epsilon$-far always from $\ket{\Psi}$, i.e.,
\begin{align}
    \cS := \{\sigma_i: \bra{\Psi}\sigma_i\ket{\Psi} \leq 1 - \epsilon\},
\end{align}
and $\overline{\cS}$ to be the complement of $\cS$, which is the set of states that are $\epsilon$-close to $\ket{\Psi}$. See
Supplementary Figure~\ref{fig:sets} for illustration.
\begin{figure}[!htbp]
  \centering
  \includegraphics[width=0.4\textwidth]{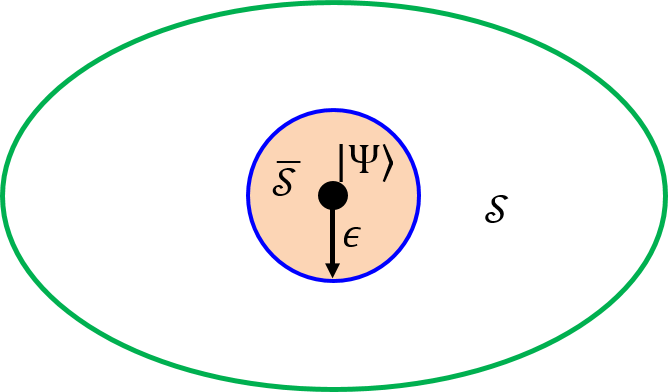}
  \caption{$\cS$ is the set of states that are at least $\epsilon$-far always from $\ket{\Psi}$, that is, states that are out
          of the brown circle. The states that are on the circle (the blue circle, $\bra{\Psi}\sigma\ket{\Psi} = 1
          -\epsilon$) are the most difficult to be distinguished from $\ket{\Psi}$, compared to states in $\cS$.}
  \label{fig:sets}
\end{figure}

Before going to the details, we compare briefly the differences between \textbf{Task A} and \textbf{Task B}.
\begin{enumerate}
  \item \textbf{Task A}~\cite{SPhysRevLett.120.170502,Swang2019optimal} considers the problem: 
        Performing tests sequentially, how many tests are required before we identify a $\sigma_i$ that does not pass
        the test? Once we find a $\sigma_i$ not passing the test, we conclude that the device is bad. Each test is
        a Bernoulli trial with probability $\Delta_\epsilon$ determined by the chosen verification strategy. The
        expected number of tests required is given by $1/\Delta_\epsilon$. It should be stressed that this expectation
        is achieved with the assumption that all generated states are $\epsilon$-far from $\ket{\Psi}$. Once the
        generated states are much further from $\ket{\Psi}$ by $\epsilon$ (\textbf{Case 2}), then the observed number
        of tests will be less than $1/\Delta_\epsilon$.
  \item \textbf{Task B}~\cite{Syu2019optimal} considers the problem: Performing $N$ tests, how many states $\sigma_i$
        can pass these tests on average? If all generated states are \textit{exactly} $\epsilon$-far from $\ket{\Psi}$,
        then the expected number of states passing the test is $N\mu$, where $\mu$ is the expected passing
        probability of each generated state, determined by the chosen verification strategy. If all generated states
        are much far from $\ket{\Psi}$ by $\epsilon$ (\textbf{Case 2}), then the expected number of states passing the
        test is less than $N\mu$. If all generated states are much close to $\ket{\Psi}$ by $\epsilon$ (\textbf{Case
        1}), then the expected number of states passing the test is larger than $N\mu$.
\end{enumerate}

\subsection{Applying the nonadaptive strategy to Task B}

Here we illustrate how to use the nonadaptive verification strategy $\Omega_{\opn{opt}}$ proposed in~\ref{sec:nonadaptive} to execute \textbf{Task B}. First of all, the following two lemmas are essential, which
state that for arbitrary state belonging to $\cS$ ($\overline{\cS}$), its probability of passing the test
$\Omega_{\opn{opt}}$ is upper (lower) bounded.

\begin{lemma}\label{lemma:pass-prob}
For arbitrary state $\sigma\in\cS$, it can pass $\Omega_{\opn{opt}}$ with probability no larger than $1 -
[1-\lambda_2(\Omega_{\opn{opt}})]\epsilon$, where $\lambda_2(\Omega_{\opn{opt}})$ is the second largest eigenvalue of
$\Omega_{\opn{opt}}$ given in Eq.~\eqref{eq:Omega-2-non}. The upper bound is achieved by states in $\cS$ that are exactly
$\epsilon$-far from $\ket{\Psi}$.
\end{lemma}
Intuitively, the further a state $\sigma\in\cS$ is from $\ket{\Psi}$, the smaller the probability it passes the test
$\Omega_{\opn{opt}}$. Lemma~\ref{lemma:pass-prob} justifies this intuition and shows quantitatively that the passing
probability for states in $\cS$ cannot be larger than $1 - [1-\lambda_2(\Omega_{\opn{opt}})]\epsilon$. However, we remark that it
is possible that for some states in $\cS$ the passing probability is small.
\begin{proof}
From the spectral decomposition Eq.~\eqref{eq:Omega-non-sd} of $\Omega_{\opn{opt}}$ and the fact that the off-diagonal
parts of $\sigma$ do not affect the trace $\tr[\Omega_{\opn{opt}}\sigma]$, we can assume without loss of
generality that $\sigma$ has the form
\begin{align}
  \sigma = p_1\proj{\Psi} + p_2\proj{\Psi^\perp} + p_3\proj{HV} + p_4\proj{VH},\;\sum_{i=1}^4 p_i = 1,\;
  p_1 \leq 1 - \epsilon,
\end{align} 
where the last constraint $p_1 \leq 1 - \epsilon$ follows from the precondition that $\sigma\in\cS$. Then
\begin{align}
  \tr\left[\Omega_{\opn{opt}}\sigma\right]
&= p_1 + \sum_{i=2}^4\lambda_ip_i \nonumber\\
&\leq p_1 + (1-p_1)\lambda_2 \nonumber\\
&\leq (1-\epsilon)(1 - \lambda_2) + \lambda_2 \nonumber\\
&= 1 - (1 - \lambda_2)\epsilon,
\end{align}
where $\lambda_i$ is the $i$-th eigenvalue of $\Omega_{\opn{opt}}$, the first inequality follows from
$\lambda_2\geq\lambda_3\geq\lambda_4$ and the second inequality follows from $p_1 \leq 1 - \epsilon$.
\end{proof}

\begin{lemma}\label{lemma:pass-prob-2}
For arbitrary state $\sigma\in\overline{\cS}$, it can pass the nonadaptive strategy $\Omega_{\opn{opt}}$ with
probability no smaller than $1 - [1-\lambda_2(\Omega_{\opn{opt}})]\epsilon$. The lower bound is achieved by states in
$\overline{\cS}$ that are exactly $\epsilon$-far from $\ket{\Psi}$.
\end{lemma}
\begin{proof}
Following the proof for Lemma~\ref{lemma:pass-prob} and considering the spectral decomposition Eq.~\eqref{eq:Omega-non-sd}
of $\Omega_{\opn{opt}}$, we can assume without loss of generality that $\sigma$ has the form
\begin{align}
  \sigma =  p_1\proj{\Psi} + p_2\proj{\Psi^\perp} + p_3\ket{VH} + p_4\ket{HV},\;\sum_{i=1}^4 p_i = 1,\;
            p_1 \geq 1 - \epsilon,
\end{align} 
where the last constraint $p_1 \geq 1 - \epsilon$ follows from the precondition that $\sigma\in\overline{\cS}$.
Then
\begin{align}
  \tr\left[\Omega_{\opn{opt}}\sigma\right]
=&\; p_1 + \lambda_2(\Omega_{\opn{opt}})\sum_{i=2}^4p_i \nonumber\\
=&\; p_1 + \lambda_2(\Omega_{\opn{opt}})(1-p_1) \nonumber\\
\geq&\; (1-\epsilon)(1 - \lambda_2(\Omega_{\opn{opt}})) + \lambda_2(\Omega_{\opn{opt}}) \nonumber\\
=&\; 1 - (1 - \lambda_2(\Omega_{\opn{opt}}))\epsilon,
\end{align}
where the first equality follows from the spectral decomposition of Eq.~\eqref{eq:Omega-non-sd} and the inequality follows from
 $p_1 \geq 1 - \epsilon$.
\end{proof}

\textbf{Task B} is to distinguish among the following \textit{two} cases for a given quantum device $\cD$:
\begin{description}
  \item[\textbf{Case 1}] $\forall i$, $\sigma_i\in\overline{\cS}$.
                          That is, the device always generate states that are 
                          sufficient close to the target state
                          $\ket{\Psi}$. If it is the case, we regard the device as ``good''.
  \item[\textbf{Case 2}] $\forall i$, $\sigma_i\in\cS$. 
                          That is, the device always generate states that are 
                          sufficient far from the target state
                          $\ket{\Psi}$. If it is the case, we regard the device as ``bad''.
\end{description}

To execute \textbf{Task B}, we perform $N$ tests on the outputs of the device with $\Omega_{\opn{opt}}$ and record the number of tests that pass the
test as $m_{\text{pass}}$. What conclusion can we obtain from the relation between $m_{\text{pass}}$ and $N$? From
Lemma~\ref{lemma:pass-prob}, we know that if the device is bad (belonging to \textbf{Case 2}), $m_\text{pass}$ cannot be
too large, since the passing probability of each generated state is upper bounded. Conversely,
Lemma~\ref{lemma:pass-prob-2} guarantees that if the device is good (belonging to \textbf{Case 1}), $m_{\text{pass}}$ cannot
be too small, since the passing probability of each generated state is lower bounded. Let's then justify this
intuition rigorously. We define the binary random variable $X_i$ to represent the event that whether state $\sigma_i$
passes the test or not. If it passes, we set $X_i=1$; If it fails to pass the test, we set $X_i=0$. After $N$ tests, we
obtain a sequence of independent distribution random variables $\{X_i\}_{i=1}^N$. Then $m_\text{pass} =
\sum_{i=1}^N X_i$. Now let's analyze the expectation of each $X_i$. Let
$\mu:=1-[1-\lambda_2(\Omega_{\opn{opt}})]\epsilon$. If the device were `good' (it belongs to \textbf{Case 1}), then
Lemma~\ref{lemma:pass-prob-2} implies that the expectation of $X_i$, denoted as $\mathbb{E}(X_i)$, shall satisfy
$\mathbb{E}(X_i) \geq \mu$. The independent assumption together with the law of large numbers guarantee $m_\text{pass}\geq N\mu$,
when $N$ is sufficiently large. On the other hand, if the device were `bad' (it belongs to \textbf{Case 2}), then
Lemma~\ref{lemma:pass-prob} asserts that the expectation of $X_i$ shall satisfy $\mathbb{E}(X_i) \leq \mu$. Again, the
independent assumption together with the law of large numbers guarantee $m_\text{pass}\leq N\mu$, when $N$ is sufficiently large.
That is to say, we consider two regions regarding the value of $m_\text{pass}$: the region that $m_\text{pass} \leq N\mu$ and the
region that $m_\text{pass} \geq N\mu$. We refer to Supplementary Figure~\ref{fig:m-to-n} for illustration.
\begin{figure}[!htbp]
  \centering
  \includegraphics[width=0.6\textwidth]{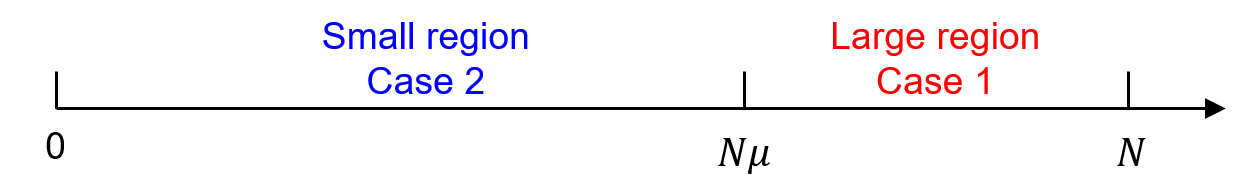}
  \caption{The number of copies $m_\text{pass}$ that pass the nonadaptive strategy $\Omega_{\opn{opt}}$ determines the
  quality of the device. If $m_\text{pass} \leq N\mu$, the device is very likely to be \textbf{Case 2} (\textbf{Small
  Region}). Conversely, the device is very likely to be \textbf{Case 1} (\textbf{Large Region}) if $m_\text{pass} \geq
  N\mu$.}
  \label{fig:m-to-n}
\end{figure}

\textbf{Large Region: $m_\text{pass} \geq N\mu$.} In this region, the device belongs to \textbf{Case 1} with high
probability, as only in this case $m_\text{pass}$ can be large. Following the analysis methods in
refs.~\cite{SDimi2018,SSaggio2019,Syu2019optimal,Szhang2019experimental}, we use the Chernoff bound to upper bound the probability that
the device belongs to \textbf{Case 2} as
\begin{align}\label{Large}
\delta\equiv e^{-N\operatorname{D}\left(\frac{m_\text{pass}}{N}\middle\lVert\mu\right)},
\end{align}
where $\operatorname{D}\left(x\lVert y\right):=x\log_2\frac{x}{y}+(1-x)\log_2\frac{1-x}{1-y}$ the Kullback-Leibler
divergence. That is to say, if the device \textit{did} belong to \textbf{Case 2}, the probability that $m_\text{pass}$ is
larger than $N\mu$ decays exponentially as $m_\text{pass}$ becomes larger. In this case, we reach a conclusion:
\begin{center}
{\it The device belongs to \textbf{Case 2} with probability at most $\delta$.}
\end{center}
Equivalently, we conclude that
\begin{center}
{\it The device belongs to \textbf{Case 1} with probability at least $1 - \delta$.}
\end{center}

\textbf{Small Region: $m_\text{pass} \leq N\mu$.} The analysis for this region is almost the same as that of the large
region. In this region, the device belongs to \textbf{Case 2} with high probability, as only in this case $m_\text{pass}$
can be smaller than $N\mu$. Following the analysis methods in
refs.~\cite{SDimi2018,SSaggio2019,Syu2019optimal,Szhang2019experimental}, we use the Chernoff bound to upper bound the probability that
the device belongs to \textbf{Case 1} as
\begin{align}\label{small}
\delta\equiv e^{-N\operatorname{D}\left(\frac{m_\text{pass}}{N}\middle\lVert\mu\right)}.
\end{align}
That is to say, if the device \textit{did} belong to \textbf{Case 1}, the probability that $m_\text{pass}$ is less than
$N\mu$ decays exponentially as $m_\text{pass}$ becomes smaller. In this case, we reach a conclusion:
\begin{center}
{\it The device belongs to \textbf{Case 1} with probability at most $\delta$.}
\end{center}
Equivalently, we conclude that
\begin{center}
{\it The device belongs to \textbf{Case 2} with probability at least $1 - \delta$.}
\end{center}

\subsection{Applying the adaptive strategy to Task B}

The way that we use the adaptive verification strategy $\Omega^\rightarrow_{\opn{opt}}$ proposed in~\ref{sec:adaptive} to execute \textbf{Task B} is similar to that discussed in the previous section. The following
two lemmas are required, which state that for arbitrary state belonging to $\cS$ ($\overline{\cS}$), its probability of
passing the test $\Omega^\rightarrow_{\opn{opt}}$ is upper (lower) bounded.

\begin{lemma}\label{lemma:pass-prob-adp}
For arbitrary state $\sigma\in\cS$, it can pass $\Omega^\rightarrow_{\opn{opt}}$ with probability no larger than
$1 - [1-\lambda_2(\Omega^\rightarrow_{\opn{opt}})]\epsilon$, where $\lambda_2(\Omega^\rightarrow_{\opn{opt}})$ is the
second largest eigenvalue of $\Omega^\rightarrow_{\opn{opt}}$ given in Eq.~\eqref{eq:Omega-2-adp}. The upper bound is
achieved by states in $\cS$ that are exactly $\epsilon$-far from $\ket{\Psi}$.
\end{lemma}

The proof of Lemma~\ref{lemma:pass-prob-adp} is the same as that of Lemma~\ref{lemma:pass-prob}, so we omit the details.

\begin{lemma}\label{lemma:pass-prob-adp-converse}
For arbitrary state $\sigma\in\overline{\cS}$, it can pass the adaptive strategy $\Omega^\rightarrow_{\opn{opt}}$ with
probability no smaller than $1 - [1-\lambda_4(\Omega^\rightarrow_{\opn{opt}})]\epsilon$, where
$\lambda_4(\Omega^\rightarrow_{\opn{opt}})=\sin^2\theta/(1+\cos^2\theta)$ is the fourth (also the smallest) eigenvalue
of $\Omega^\rightarrow_{\opn{opt}}$. The lower bound is achieved by states in $\overline{\cS}$ that are exactly
$\epsilon$-far from $\ket{\Psi}$.
\end{lemma}
\begin{proof}
Following the proof for Lemma~\ref{lemma:pass-prob-2} and considering the spectral decomposition Eq.~\eqref{eq:Omega-adp-sd}
of $\Omega^\rightarrow_{\opn{opt}}$, we can assume without loss of generality that $\sigma$ has the form
\begin{align}
  \sigma =  p_1\proj{\Psi} + p_2\proj{\Psi^\perp} + p_3\ket{VH} + p_4\ket{HV},\;\sum_{i=1}^4 p_i = 1,\;
            p_1 \geq 1 - \epsilon,
\end{align} 
where the last constraint $p_1 \geq 1 - \epsilon$ follows from the precondition that $\sigma\in\overline{\cS}$. Then
\begin{align}
  \tr\left[\Omega^\rightarrow_{\opn{opt}}\sigma\right]
=&\; p_1  + \sum_{i=2}^4\lambda_i(\Omega^\rightarrow_{\opn{opt}})p_i \nonumber\\
\geq&\; p_1 + \lambda_4(\Omega^\rightarrow_{\opn{opt}})\sum_{i=2}^4p_i \nonumber\\
=&\; p_1  + \lambda_4(\Omega^\rightarrow_{\opn{opt}})(1-p_1) \nonumber\\
\geq&\; (1-\epsilon)(1 - \lambda_4(\Omega^\rightarrow_{\opn{opt}})) + \lambda_4(\Omega^\rightarrow_{\opn{opt}}) \nonumber\\
=&\; 1 - (1 - \lambda_4(\Omega^\rightarrow_{\opn{opt}}))\epsilon,
\end{align}
where the first inequality follows from the fact that $\lambda_4(\Omega^\rightarrow_{\opn{opt}})$
is the smallest eigenvalue of $\Omega^\rightarrow_{\opn{opt}}$ and the second inequality follows from
 $p_1 \geq 1 - \epsilon$.
\end{proof}

To execute \textbf{Task B}, we perform $N$ tests on the outputs of the device with $\Omega^\rightarrow_{\opn{opt}}$ and record the number of tests that pass the
test as $m_\text{pass}$. Define the binary random variable $X_i$ to present the event that whether state $\sigma_i$
passes the test or not. If it passes, we set $X_i=1$; If it fails to pass the test, we set $X_i=0$. After $N$ tests, we
obtain a sequence of independent random variables $\{X_i\}_{i=1}^N$. Then $m_\text{pass} =
\sum_{i=1}^N X_i$. We analyze the expectation of each $X_i$. Let
$\mu_l:=1-[1-\lambda_2(\Omega^\rightarrow_{\opn{opt}})]\epsilon$ and
$\mu_s:=1-[1-\lambda_4(\Omega^\rightarrow_{\opn{opt}})]\epsilon$. 
If the device were `good' (it belongs to \textbf{Case 1}), then
Lemma~\ref{lemma:pass-prob-adp-converse} implies that 
$\mathbb{E}(X_i) \geq \mu_s$. 
The independent assumption together with the law of large numbers guarantee $m_\text{pass}\geq N\mu_s$,
when $N$ is sufficiently large. 
That is to say, if in practical it turns out that $m_\text{pass}\leq N\mu_s$, then we are pretty sure
the device is bad. 
On the other hand, if the device were `bad' (it belongs to \textbf{Case 2}), then
Lemma~\ref{lemma:pass-prob-adp} asserts that $\mathbb{E}(X_i) \leq \mu_l$. Again, the
independent assumption together with the law of large numbers guarantee $m_\text{pass}\leq N\mu_l$, when $N$ is sufficiently large.
That is, if in practical it turns out that $m_\text{pass}\geq N\mu_l$, then we are pretty sure
the device is good. 
Thus, we shall consider two regions regarding the value of $m_\text{pass}$: the region that $m_\text{pass} \leq N\mu_s$ and the
region that $m_\text{pass} \geq N\mu_l$. We refer to Supplementary Figure~\ref{fig:m-to-n-adp} for illustration.
\begin{figure}[!htbp]
  \centering
  \includegraphics[width=0.6\textwidth]{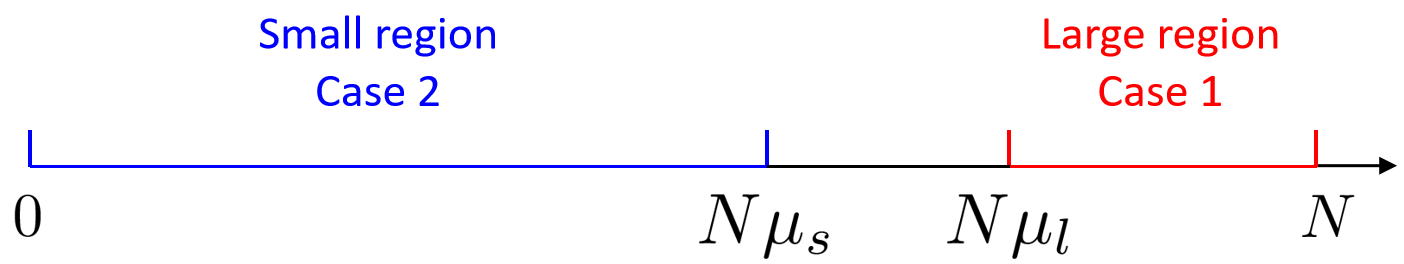}
  \caption{The number of copies $m_\text{pass}$ that pass the adaptive strategy $\Omega^\rightarrow_{\opn{opt}}$ determines the
  quality of the device. If $m_\text{pass} \leq N\mu_s$, the device belongs very likely to \textbf{Case 2} (\textbf{Small
  Region}). Conversely, the device belongs with high probability to \textbf{Case 1} (\textbf{Large Region}) if
  $m_\text{pass} \geq N\mu_l$.}
  \label{fig:m-to-n-adp}
\end{figure}

\textbf{Large Region: $m_\text{pass} \geq N\mu_l$.} In this region, the device belongs to \textbf{Case 1} with high
probability. Following the analysis methods in refs.~\cite{SDimi2018,SSaggio2019,Syu2019optimal,Szhang2019experimental}, we use the
Chernoff bound to upper bound the probability that the device belongs to \textbf{Case 2} as
\begin{align}\label{Large-adp}
\delta_l \equiv e^{-N\operatorname{D}\left(\frac{m_\text{pass}}{N}\middle\lVert\mu_l\right)}.
\end{align} 
That is to say, if the device \textit{did} belong to \textbf{Case 2}, the probability that $m_\text{pass}$ is larger than
$N\mu$ decays exponentially when $m_\text{pass}$ becomes large. In this case, we reach a conclusion:
\begin{center}
{\it The device belongs to \textbf{Case 2} with probability at most $\delta_l$.}
\end{center}
Equivalently, we conclude that
\begin{center}
{\it The device belongs to \textbf{Case 1} with probability at least $1 - \delta_l$.}
\end{center}

\textbf{Small Region: $m_\text{pass} \leq N\mu_s$.} The analysis for this region is almost the same as that of the large
region. In this region, the device belongs to \textbf{Case 2} with high probability, as only in this case $m_\text{pass}$
can be smaller than $N\mu_s$. Following the analysis methods in refs.~\cite{SDimi2018,SSaggio2019,Syu2019optimal,Szhang2019experimental},
we use the Chernoff bound to upper bound the probability that the device belongs to \textbf{Case 1} as
\begin{align}\label{small-adp}
\delta_s \equiv e^{-N\operatorname{D}\left(\frac{m_\text{pass}}{N}\middle\lVert\mu_s\right)},
\end{align}
That is to say, if the device \textit{did} belong to \textbf{Case 1}, the probability that $m_\text{pass}$ is less than
$N\mu$ decays exponentially as $m_\text{pass}$ becomes smaller. In this case, we reach a conclusion:
\begin{center}
{\it The device belongs to \textbf{Case 1} with probability at most $\delta_s$.}
\end{center}
Equivalently, we conclude that
\begin{center}
{\it The device belongs to \textbf{Case 2} with probability at least $1 - \delta_s$.}
\end{center}

\section{The experimental results for the three target states}

\noindent In this section, we give the results for three target states, i.e.,~k2, Max and HV (see Supplementary Table~\ref{truestate}). Note that the failing probability for each measurement is different for the three states~\cite{SPhysRevLett.120.170502}
\begin{subequations}\label{eq:ThreeFail}
	\begin{align}
		\text{k2}: \Delta_\epsilon &= \frac{\epsilon}{2+\sin\theta\cos\theta} \\
		\text{Max}: \Delta_\epsilon &= \frac{2\epsilon}{3} \\
		\text{HV}: \Delta_\epsilon &= \epsilon
	\end{align}
\end{subequations}
Here we discuss the results of the nonadaptive strategy. In Eq.~\eqref{eq:ThreeFail}, $\epsilon$ can be replaced with $1-F$, where $F$ is the fidelity for corresponding states. From the experimental data, the stable passing probability $m_\text{pass}/N$ can be obtained, which are 0.9986$\pm$0.0002, 0.9982$\pm$0.0002 and 0.9992$\pm$0.0001, respectively. With $\mu=1-\Delta_\epsilon=m_\text{pass}/N$, we have an estimation for the fidelities of these three states, which are $F_{\text{k2}}=0.9964\pm0.0002$, $F_{\text{Max}}=0.9973\pm0.0002$ and $F_{\text{HV}}=0.9992\pm0.0001$. The variation of $\delta$ versus $N$ for the three target states is shown in Supplementary Figure~\ref{fig:ThreeDelta} in the log scale. To make the comparison fairly, we take $\eta:=|\mu-m_\text{pass}/N|$ to be same for the three states. In this condition, the slope of decline (i.e.,~$g$ in $\log\delta \propto g\cdot N$) is $g_{HV}>g_{k2}>g_{Max}$, which indicates that the state with a larger passing probability will have a faster decline.
\begin{figure}[!htbp]
  \centering
  	\includegraphics[width=0.7\textwidth]{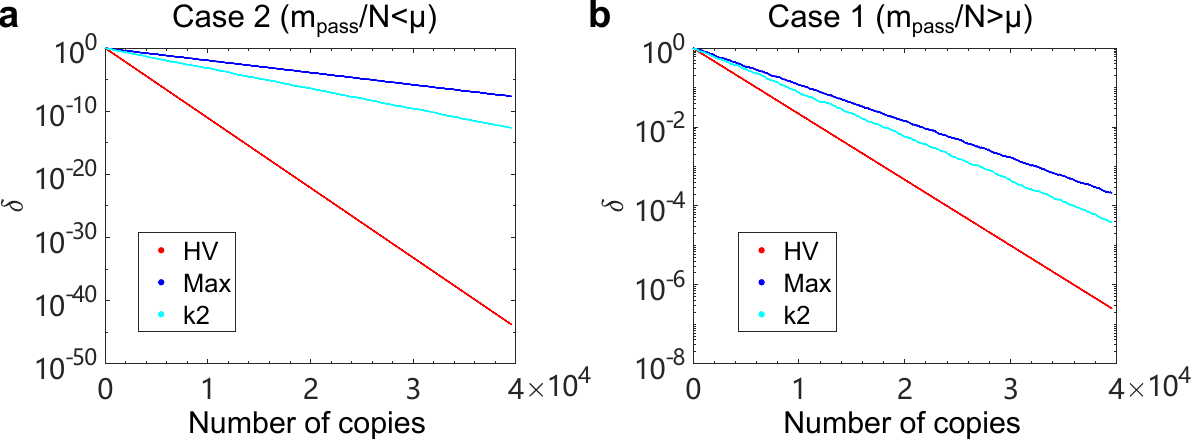}
  \caption{Experimental results on the variation of $\delta$ versus the number of copies for the three target states. Here the difference $\eta:=|\mu-m_\text{pass}/N|$ is set the same for the three states. The small region (\textbf{a}) decrease faster than the large region (\textbf{b}), which indicates it is easier to verify \textbf{Case 2} than \textbf{Case 1}. The state with a larger passing probability will have a faster slope of decline. Note that the experimental data symbols shown in the figure looks like lines due to the dense data points.}
  \label{fig:ThreeDelta}
\end{figure}

\begin{figure}[!htbp]
  \centering
  	\includegraphics[width=0.7\textwidth]{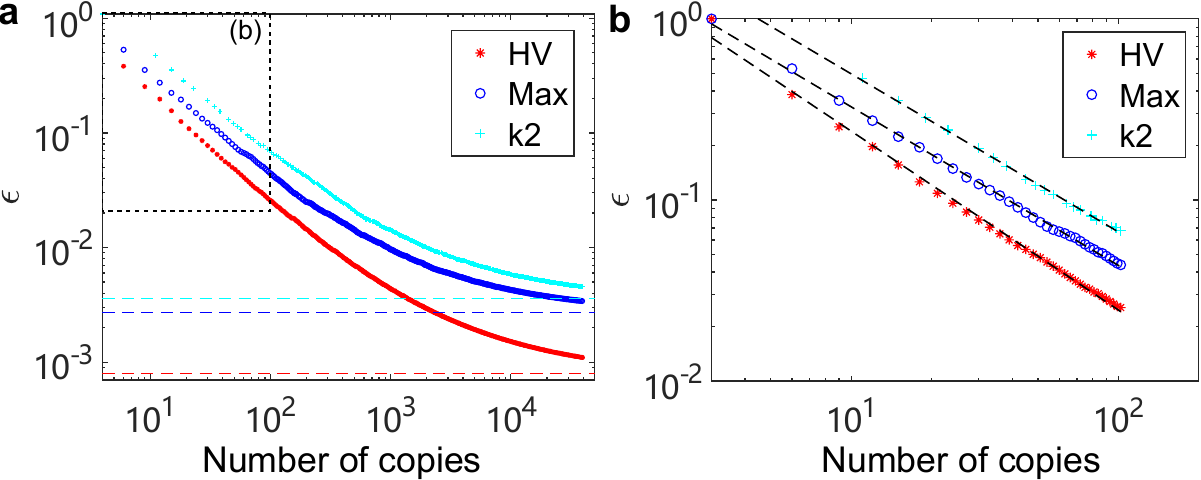}
  \caption{Experimental results on the variation of $\epsilon$ versus the number of copies for the three target states. The $\epsilon$ will finally approach the asymptotic line of the infidelities 0.0036$\pm$0.0002, 0.0027$\pm$0.0002 and 0.0008$\pm$0.0001 for these three states. \textbf{b} is the linear region of the enlarged dashed box in \textbf{a}. We fit the linear region (dashed black line) and obtain slopes of -0.88$\pm$0.03, -0.87$\pm$0.10 and -0.99$\pm$0.09 for k2, Max and HV states, respectively.}
  \label{fig:ThreeEpsilon}
\end{figure}
In Supplementary Figure~\ref{fig:ThreeEpsilon}, we present the results of $\epsilon$ versus $N$. We can see that $\epsilon$ first decreases linearly with $N$ and then approaches a asymptote in the log-log scale. The asymptotic values for these three states are the infidelities calculated from Eq.~\eqref{eq:ThreeFail}, which are 0.0036$\pm$0.0002, 0.0027$\pm$0.0002 and 0.0008$\pm$0.0001, respectively. To see the scaling of $\epsilon$ versus $N$, the region for the number of copies from 1-100 is enlarged and shown in Supplementary Figure~\ref{fig:ThreeEpsilon}b. The linear scaling region is fitted with slopes of -0.88$\pm$0.03, -0.87$\pm$0.10 and -0.99$\pm$0.09, respectively. The error bars are obtained by fitting different groups of $\epsilon$ versus $N$ when considering the errors of $\epsilon$ (shown in Fig.~5 of the main text).

\section{Explanations for the results in Task A}
\noindent In this section, we explain on the difference of experimental and theoretical improvements for the number of measurements required in \textbf{Task A} of main text. Consider a general state $\rho$ produced in the experiment that is exactly $\epsilon$-far from $\ket{\Psi}$, which can
be diagonalized in the bases $\{\ket{\Psi},\ket{\Psi^\perp},\ket{HV},\ket{VH}\}$ as:
\begin{align}
    \rho = (1-\epsilon)\ket{\Psi}\bra{\Psi}+p_2\ket{\Psi^\perp}\bra{\Psi^\perp}+p_3\ket{HV}\bra{HV}+p_4\ket{VH}\bra{VH},
\end{align}
where the normalization condition requires $p_2+p_3+p_4=\epsilon$.
Given the spectral decomposition of $\Omega_{\opn{opt}}$ in Eq.~\eqref{eq:Omega-non-sd},
the passing probability of $\rho$ that passes the nonadaptive strategy $\Omega_{\opn{opt}}$ can be expressed as
\begin{align}\label{eq:passNon}
  \tr\left[\Omega_{\opn{opt}}\rho\right]
&= 1-\epsilon + \lambda_2(\Omega_{\opn{opt}})(p_2+p_3+p_4) \nonumber \\
&= 1- [1 - \lambda_2(\Omega_{\opn{opt}})]\epsilon \nonumber \\
&= 1- \frac{1}{2+\sin\theta\cos\theta}\epsilon \nonumber \\
&\equiv 1 - \Delta_\epsilon.
\end{align}
From Eq.~\eqref{eq:passNon} we can see that the passing probability of nonadaptive strategy is independent on $p_2$,
$p_3$ and $p_4$.

Likewise, the passing probability of the adaptive strategy $\Omega^\rightarrow_{\opn{opt}}$ can be expressed as,
\begin{align}\label{eq:passAdp}
  \tr\left[\Omega^\rightarrow_{\opn{opt}}\rho\right]
&= 1-\epsilon + \lambda_2(\Omega^\rightarrow_{\opn{opt}})(p_2+p_3)
              + \lambda_4(\Omega^\rightarrow_{\opn{opt}})p_4 \nonumber \\
&= 1-\epsilon + \lambda_2(\Omega^\rightarrow_{\opn{opt}})(p_2+p_3+p_4)
   - \left(\lambda_2(\Omega^\rightarrow_{\opn{opt}}) - \lambda_4(\Omega^\rightarrow_{\opn{opt}})\right)p_4\nonumber \\
&= 1 - \left(1 - \lambda_2(\Omega^\rightarrow_{\opn{opt}})\right)\epsilon
   - \left(\lambda_2(\Omega^\rightarrow_{\opn{opt}}) - \lambda_4(\Omega^\rightarrow_{\opn{opt}})\right)p_4\nonumber \\
&= 1 - \frac{1}{1 + \cos^2\theta}\epsilon - \frac{\cos^2\theta - \sin^2\theta}{1 + \cos^2\theta}p_4\nonumber \\
&\equiv 1-\Delta_\epsilon^\rightarrow.
\end{align}
We can see that the passing probability of adaptive is not only dependent on the $\epsilon$, but also dependent on $p_4$. Note that the number of measurements required to obtain a $1-\delta$ confidence and infidelity $\epsilon$ is given by $n\propto\frac{1}{\Delta_\epsilon}\ln\frac{1}{\delta}$. Therefore, the improvement of adaptive relative to nonadaptive is,
\begin{align}
	\frac{n^\text{Non}}{n^\text{Adp}}\propto\frac{\Delta_\epsilon^\rightarrow}{\Delta_\epsilon}.
\end{align}
From Eqs.~\eqref{eq:passNon} and~\eqref{eq:passAdp}, we can see that the ratio of the coefficient of $\epsilon$, i.e.,~$(2+\sin\theta\cos\theta)$:$(1+\cos^2\theta)$,
is just the theoretical prediction, 
which is about 1.6 times for the $k2$ state. The remaining term is dependent on $p_4$ in Eq.~\eqref{eq:passAdp}. On the other hand, the experimental limited fidelity of the EOMs' modulation will result in a lower passing probability of the adaptive measurement. This also leads to a larger $\Delta_\epsilon^\rightarrow$ in Eq.~\eqref{eq:passAdp}, which improves the ratio. Based on our experimental data, we give a quantitative estimation. The failing probability of nonadaptive is about $\Delta_\epsilon\sim0.0014$ (see Supplementary Table~\ref{truestate}). If the passing probability of adaptive is decreased by 0.007, the number of copies will have a 0.007/0.0014$\sim$5 times reduction due to such a small denominator of $\Delta_\epsilon$. This leads to the overall about six times fewer number of copies cost by the adaptive strategy compared with the nonadaptive strategy.

\section{Comparison of nonadaptive and adaptive in Task B}
\noindent In Supplementary Figure~\ref{fig:mN}, we give the variation of experimental passing probability $m_\text{pass}/N$ versus the number of copies. For clearness, we also plot the expected passing probability $\mu=1-\Delta_\epsilon$ which is chosen by the verifier. The $\epsilon_\text{min}$ is adopted for \textbf{Case 2} and the $\epsilon_\text{max}$ is adopted for \textbf{Case 1}. With enough number of copies of quantum states, the passing probability $m_\text{pass}/N$ will reach a stable value. It can be seen that the experimental passing probability is smaller than the expectation for \textbf{Case 2} while it is larger for \textbf{Case 1}. The stable passing probability of adaptive 0.9914$\pm$0.0005 is smaller than the nonadaptive strategy 0.9986$\pm$0.0002 due to their different infidelities.
\begin{figure}[!htbp]
  \centering
  	\includegraphics[width=0.7\textwidth]{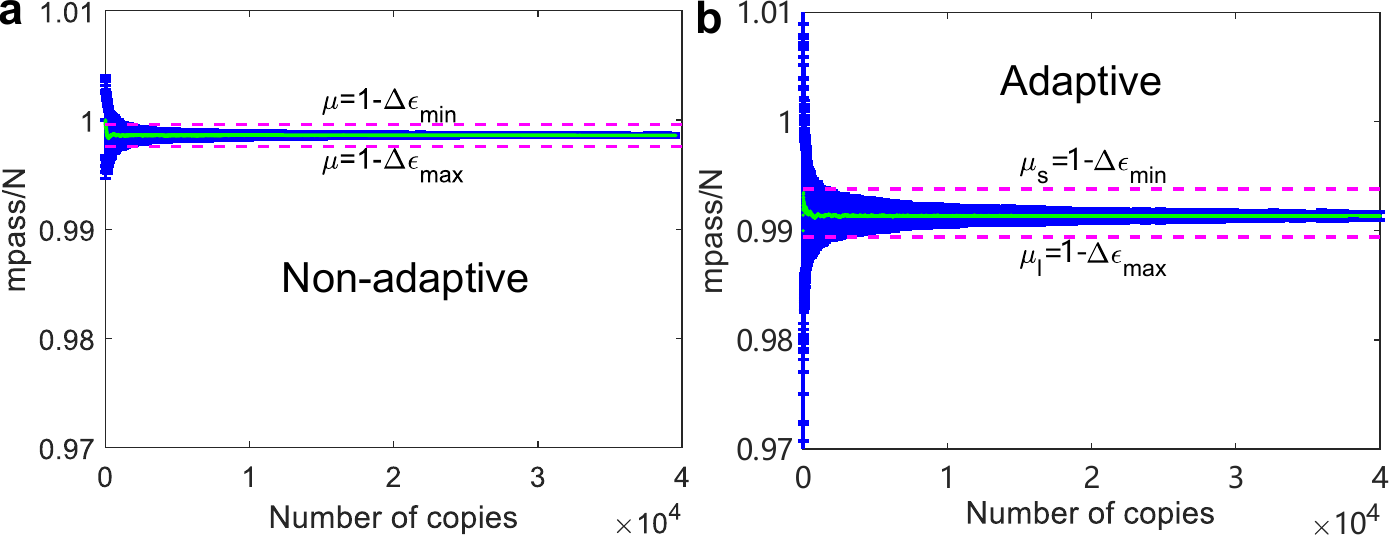}
  \caption{Experimental results on the variation of passing probability $m_\text{pass}/N$ versus number of copies for $\epsilon=\epsilon_\text{min}$ and $\epsilon=\epsilon_\text{max}$. \textbf{a}~are for the nonadaptive strategy, while \textbf{b}~are for the adaptive strategy. The dashed magenta line is the corresponding expected passing probability $\mu=1-\Delta\epsilon_\text{min}$/$\mu=1-\Delta\epsilon_\text{max}$ chosen by the verifier for the nonadaptive strategy. For adaptive strategy, the $\epsilon_\text{min}$ ($\epsilon_\text{max}$) is chosen so that $m_\text{pass}/N<\mu_s$ ($m_\text{pass}/N>\mu_l$). The experimental passing probability $m_\text{pass}/N$ reaches a stable value after about 1000 number of copies. The adaptive $m_\text{pass}/N$ is smaller than the nonadaptive. The blue is the experimental error bar, which is obtained by 100 rounds for each copy.}
  \label{fig:mN}
\end{figure}

In the main text, we see that the scaling behaviour of parameters $\delta$ and $\epsilon$ versus number of copies $N$ is different for nonadaptive and adaptive strategies in \textbf{Task B}. Here, we give some analyses to explain.

For the variation of $\delta$ versus $N$, the speed of descent is determined by the difference between experimental passing probability $m_\text{pass}/N$ and expected passing probability $\mu=1-\Delta_\epsilon$. Because the nonadaptive and adaptive strategy have different $m_\text{pass}/N$ and $\mu$, the behaviour of $\delta$ versus $N$ is also different. To have a comprehensive understanding, we consider two situations. First, we assume they have the same $m_\text{pass}/N$ and the same $\epsilon$. The expected passing probabilities are $\mu_\text{Non}=1-\epsilon/(2+\sin\theta\cos\theta)$ and $\mu_\text{Adp}=1-\epsilon/(2-\sin^2\theta)$ for nonadaptive and adaptive strategies, respectively. We can see that $\mu_\text{Non}$ is larger than $\mu_\text{Adp}$ under the same $\epsilon$. In Supplementary Figure~\ref{fig:SamemN}, we plot the variation of $\delta$ versus $\mu$ using equation $\delta = e^{-N\operatorname{D}\left(m_\text{pass}/N\middle\lVert\mu\right)}$, irrespective of the specific strategy. Therefore, the curves of nonadaptive and adaptive strategy coincide with each other in this situation. The expected probabilities $\mu_\text{Non}$ and $\mu_\text{Adp}$ are located at the different positions of the horizontal axis, as shown in the inset of Supplementary Figure~\ref{fig:SamemN}. Note that $\mu{\leq}m_\text{pass}/N$ is for \textbf{Case 1} and $\mu{\geq}m_\text{pass}/N$ is for \textbf{Case 2} region when we choose different $\mu$, i.e.,~$\epsilon$ ($\mu=1-\Delta_\epsilon$). From the figure we can see that the larger the difference between $\mu$ and $m_\text{pass}/N$, the smaller of the $\delta$. This indicates that $\delta$ decreases more quickly for adaptive strategy in the \textbf{Case 1} region, whereas it decreases more quickly for nonadaptive strategy in the \textbf{Case 2} region.
\begin{figure}[!htbp]
  \centering
  	\includegraphics[width=0.5\textwidth]{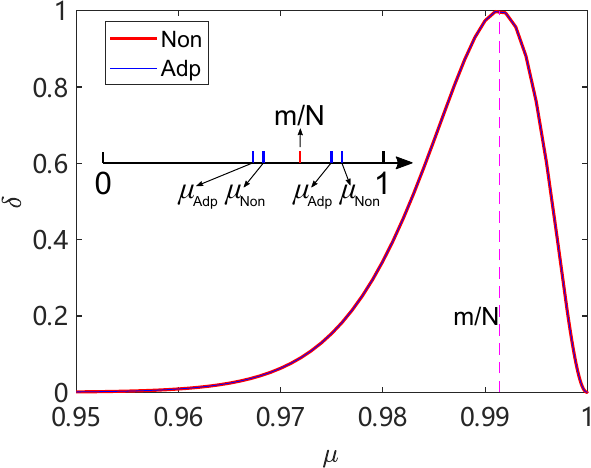}
  \caption{The parameter $\delta$ changes with the anticipating passing probability $\mu$ when adopting the same $m_\text{pass}/N$. With the same $\epsilon$, adaptive will always have a smaller $\mu$ than nonadaptive. Therefore, the adaptive will decrease faster than nonadaptive due to a larger difference in the $\mu{\leq}m_\text{pass}/N$ region. In the $\mu{\geq}m_\text{pass}/N$ region, the trend is opposite.}
  \label{fig:SamemN}
\end{figure}

Second, we consider the situation where the nonadaptive and adaptive strategies have different passing probabilities $m_\text{pass}/N$, as in our experiment. For comparison, we adopt $m_\text{pass}/N=0.9986$ for nonadaptive and $m_\text{pass}/N=0.9914$ for adaptive based on our experimental data. The variation of $\delta$ versus $\mu$ is shown in Supplementary Figure~\ref{fig:DifmN}. We can see that the comparison of nonadaptive and adaptive depends on the choice of $\mu$. In the figure, we label the crosspoint of the two curves as $\mu_0$. If the $\epsilon$ is chosen such that both $\mu_\text{Non}$ and $\mu_\text{Adp}$ are smaller than $\mu_0$, the nonadaptive will drop faster than adaptive because nonadaptive has a smaller $\delta$. On the contrary, if both $\mu_\text{Non}$ and $\mu_\text{Adp}$ are larger than $\mu_0$, the adaptive will drop faster than nonadaptive due to the smaller $\delta$. If $\mu_\text{Non}$ and $\mu_\text{Adp}$ are located at two sides of $\mu_0$, the one which has a smaller $\delta$ in Supplementary Figure~\ref{fig:DifmN} will have a faster decline.
\begin{figure}[!htbp]
  \centering
  	\includegraphics[width=0.5\textwidth]{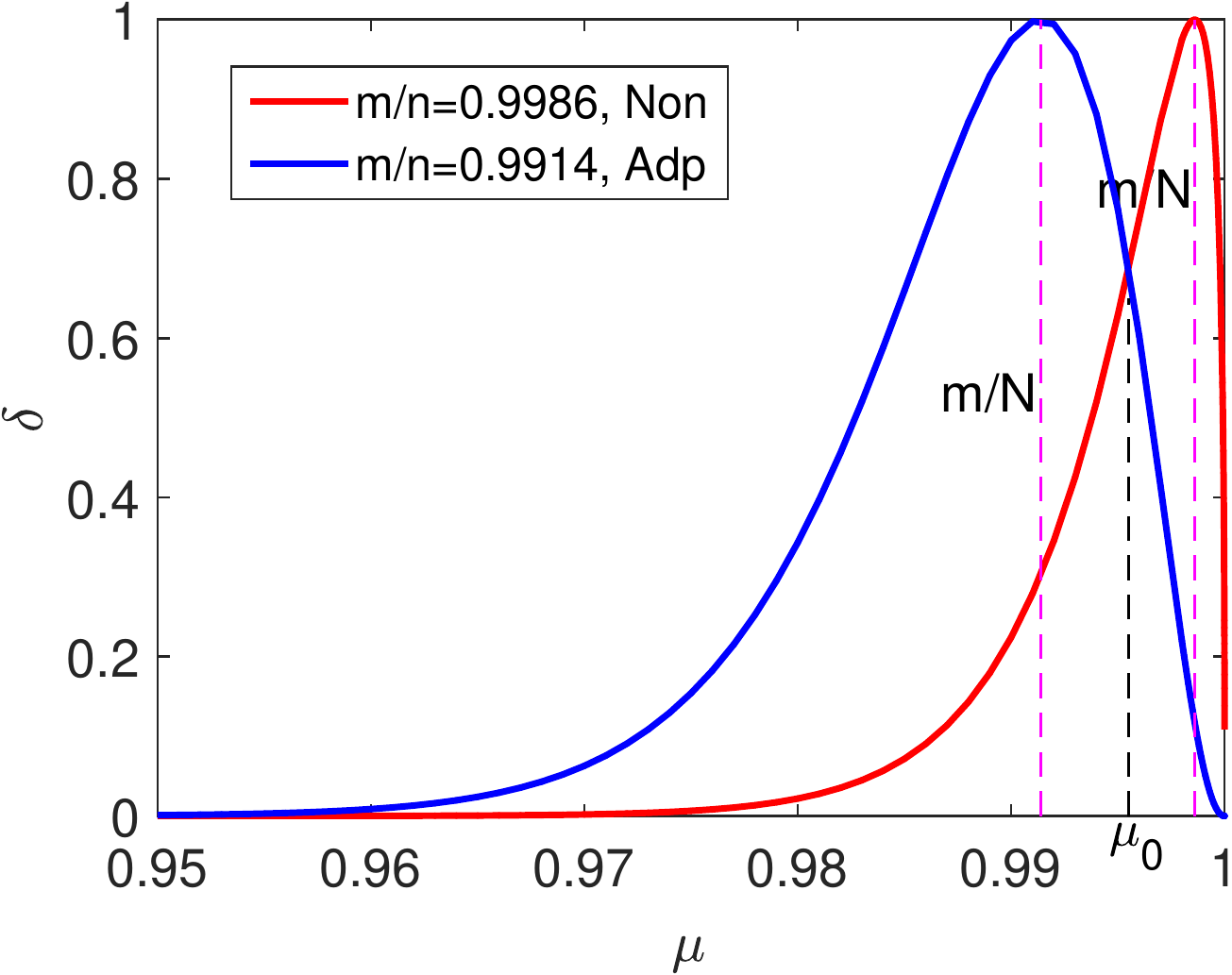}
  \caption{The parameter $\delta$ changes with the anticipating passing probability $\mu$ for different $m_\text{pass}/N$. In this situation, the one which adopts $\mu$ that results in a smaller $\delta$ will have a faster decline.}
  \label{fig:DifmN}
\end{figure}

\begin{figure}[!htbp]
  \centering
  	\includegraphics[width=0.5\textwidth]{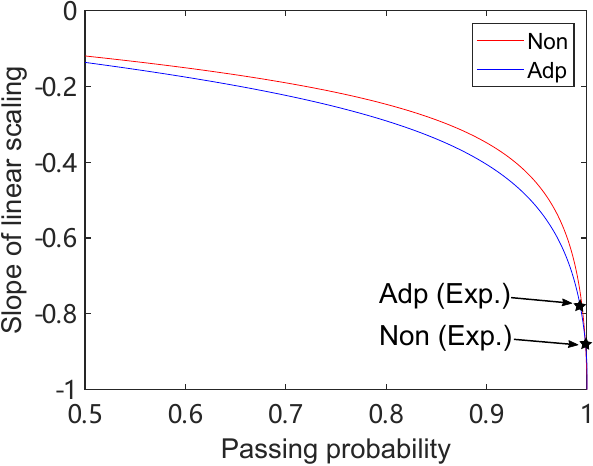}
  \caption{The slope of linear scaling versus the passing probability. The slope decreases with the increase of the passing probability. Our experimental fitting slopes for the adaptive and nonadaptive strategies in the linear region is shown as the black stars. The slope approaches the optimal -1 when the passing probability are close 1. Due to the smaller passing probability of adaptive compared with nonadaptive, the absolute value of slope for adaptive strategy is also smaller. When the passing probability is large enough, the two strategies almost have an equal linear scaling slope at the same passing probability.}
  \label{fig:Slope}
\end{figure}
For $\epsilon$ versus the number of copies $N$, the scaling is determined by the passing probability $m_\text{pass}/N$, i.e.,~the fidelity of our generated states. The larger the passing probability, the faster the infidelity parameter $\epsilon$ decreases with the number of copies. For both nonadaptive and adaptive, $\epsilon$ will finally approach a asymptotic value. In Supplementary Figure~\ref{fig:Slope}, we plot the fitting slope of the linear scaling region versus the passing probability for both the nonadaptive and adaptive strategies. We can see that the adaptive strategy will have an advantage compared with nonadaptive strategy at a small passing probability. However, the slope tends to the optimal value -1 when the passing probability is large enough. This indicates that the optimal scaling can only be obtained at a high fidelity for the generated states. In the region of high fidelity, there is minor differences for the nonadaptive and adaptive strategies if we obtain a same passing probability. However, the adaptive strategy has a smaller passing probability than the nonadaptive strategy for our experimental data. Therefore, the descent speed of adaptive strategy is slower than the nonadaptive strategy. This can be seen quantitatively from their fitting slopes -0.78$\pm$0.07 (Adp) and -0.88$\pm$0.03 (Non), shown as stars in Supplementary Figure~\ref{fig:Slope}.

\vspace*{1.0cm}
\noindent{\bfseries Supplementary References}

\end{document}